\documentclass[preprint,11pt]{elsarticle}

\usepackage{amsmath,amssymb,latexsym,bm,graphicx,lscape,color,verbatim}
\usepackage{float}
\usepackage{enumitem}
\usepackage{amsthm}

\newtheorem{Thm}{Theorem}
\newtheorem{Pro}{Proposition}
\newtheorem{Lem}{Lemma}
\newtheorem{Cla}{Claim}
\newtheorem{Cor}{Corollary}

\newtheorem{Rem}{Remark}
\newtheorem{Ass}{Assumption}

\def\qed{\hfill \vrule height 6pt width 6pt depth 0pt}
\def\etal{{\it et al.}}

\usepackage{fancybox}
\usepackage{multirow}

\usepackage{natbib}

 \bibpunct[, ]{(}{)}{,}{a}{}{,}%
 \def\bibfont{\small}%
 \def\bibsep{\smallskipamount}%
 \def\bibhang{24pt}%
 \def\newblock{\ }%
\def\eps{\varepsilon}

\DeclareMathOperator{\var}{Var}
\def\q{\quad}

\journal{Stochastic Processes and their Applications}

\begin{document}

\begin{frontmatter}



\title{Volatility Inference in the Presence of Both  Endogenous Time and Microstructure Noise}


\author[1]{Yingying Li\fnref{label1}}
\address[1]{Department of Information Systems, Business Statistics and Operations Management, Hong Kong University of Science and Technology. Email: yyli@ust.hk}

\author[2]{Zhiyuan Zhang\fnref{label2}\corref{cor1}}
\address[2]{School of Statistics and Management, Shanghai University of Finance and Economics. Email: zhang.zhiyuan@mail.shufe.edu.cn}

\author[3]{Xinghua Zheng\fnref{label3}}
\address[3]{Department of Information Systems, Business Statistics and Operations Management, Hong Kong University of Science and Technology. Email: xhzheng@ust.hk}

\fntext[label1]{Research partially supported
by GRF 602710 of the HKSAR.}
\fntext[label2]{Research partially sponsored
by Shanghai Pujiang Program 12PJC051.}
\fntext[label3]{Research partially supported
by GRF 602710 and GRF 606811 of the HKSAR.}
 \cortext[cor1]{To whom correspondence should be addressed.}

\begin{abstract}
In this article we consider the volatility inference in the presence
of both market microstructure noise and endogenous time. Estimators
of the integrated volatility in such a setting are proposed, and
their asymptotic properties are studied. Our proposed estimator is
compared with the existing popular volatility estimators via
numerical studies. The results show that our estimator can have
substantially better performance when time endogeneity exists.
\end{abstract}


\end{frontmatter}


\section{Introduction}
In recent years there has been growing interest in the inference for asset price volatilities
based on high-frequency financial data. Suppose that the latent  log price $X=(X_t)$ follows an It\^o process
\begin{align}
dX_t=\mu_t\ dt+\sigma_t\ dW_t,~{\rm for}~t\in[0,1],
\label{eq:latentP}
\end{align}
where $W$ is a standard Brownian motion, and the drift $(\mu_t)$ and volatility $(\sigma_t)$ are both stochastic processes. Econometric interests are usually in the inference for the integrated volatility, i.e., quadratic variation, of the  log price process
$$
\langle X,X\rangle_t=\int_0^t\sigma^2_s\ ds.
$$
A classical estimator from probability theory (see, for example,
\citet{jp98}, \citet{BarndShep}) for  this quantity is the realized volatility (RV)   based on the discrete time
observations
$$
X_{t_i}~{\rm for}~0=t_0<t_1<\ldots<t_{N_1}= 1,
$$
 where $t_i$'s may be a sequence of stopping times. The RV $[X,X]_t$ is defined as the sum of squared log returns
$$
[X,X]_t=\sum_{t_{i}\leq t}(\Delta X_{t_i})^2,
$$
where $\Delta X_{t_i}=X_{t_{i}}-X_{t_{i-1}}$ for $i\geq1$. Under mild conditions, when the observation frequency~$N_1$ goes to infinity, $[X,X]_t\overset{p}\longrightarrow \langle X,X\rangle_t$. Furthermore, when the observation times $(t_i)_{i\geq 0}$ are independent of~$X$, a complete asymptotic theory for the estimator $[X,X]_t$ is available, which says that $\sqrt{N_1}([X,X]_t-\langle X,X\rangle_t)$ is asymptotically a mixture of normal whose mixture component is the variance equal to $2\int_0^t\sigma^4_s\ dH_s$, where $H_t$ is the  ``quadratic variation of time''  process provided that the following limit exists (see \cite{MZ06} or  \cite{MZ2012})
$$
{\rm plim}_{N_1\rightarrow \infty}N_1\sum_{t_i\leq t}(\Delta t_i)^2=H_t,
$$
where ``plim'' stands for limit in probability.
The quantity  $\int_0^t\sigma^4_s\ dH_s$ can be consistently estimated by the quarticity $N_1/3\cdot [X,X,X,X]_t:=N_1/3\cdot\sum_{t_i\leq t}(\Delta X_{t_i})^4$.

The above provides
a foundation for estimating the integrated volatility based on high frequency data.
However, when it comes to the practical side, the assumptions for RV are often violated. Two aspects are of great importance. They are
\begin{itemize}[noitemsep,nolistsep]
\item[(a)] Market microstructure noise; \q and
\item[(b)] Endogeneity in the price sampling times.
\end{itemize}
For the first issue, recently there has seen a large literature on estimating quantities of interest with prices observed with microstructure noise. One commonly used assumption is that the noises are additive and one observes
\begin{align}
Y_{t_i}:=X_{t_i}+\eps_{t_i},~{\rm for}~i=0,1,\ldots,N_1.\label{eq:obsMSprice}
\end{align}
It is often assumed that the noise  $(\eps_{t_i})_{i\geq 1}$ is  an independent sequence of white
noise and the sampling times  $(t_i)_{i\geq 1}$ are independent of $X$.
Various estimators of integrated volatility have been proposed. See, for
example, two scales realized volatility of
\citet{ZhangMyklandAitSahalia},  multi-scale realized volatility
by \cite{Zhang},  realized kernels of \citet{NHLS},
pre-averaging method by \citet{JLMPV} and  QMLE method by
\cite{DachengXIU}.
Related works
include \cite{AMZ:2005}, \cite{bandirussell06}, \cite{fanwang07},  \citet{HL}, \cite{kalninalinton08},
\cite{limykland07}, \citet{PY2007} among others.

In contrast, issue (b) has only recently been brought to researchers' attention. The case when the sampling times
are irregular or random but (conditionally) independent of the price process has been studied by \citet{YacineMyklandA}, \citet{DG},
\citet{MRW},  \citet{HJY} among others.
A recent work of \citet{RENAULTandWERKER} provides a detailed discussion on the issue
of possible endogenous effect that stems from the price sampling
times in a semi-parametric context. \citet{lmrzz09} further
investigate the time endogeneity effect on volatility estimation in
a nonparametric setting. Volatility estimation in the presence of
endogenous time in some special situations like when the observation
times are hitting times has been studied in \cite{fukasawa10a} and
\cite{FukasawaRosenbaum}, and in a general situation has also been
studied in \citet{fukasawa10b}.
In \citet{lmrzz09},  the analysis was carried out by considering the time
endogeneity effect which is reflected by
\begin{equation}\label{eq:tricity}
{\rm plim}\sqrt{N_1}[X,X,X]_t
\end{equation}
where $\sqrt{N_1}[X,X,X]_t:=\sqrt{N_1}\sum_{t_i\leq t}(\Delta X_{t_i})^3$ is the tricity. Interestingly, the literature usually neglects the important information one could draw from the quantity $[X,X,X]_t$, which
can be interpreted as a measure of the covariance between the price process and time as shown in \citet{lmrzz09}.
\citet{lmrzz09} also conducted empirical work that provides compelling evidence that the endogenous effect does exist in financial data, i.e., ${\rm plim}\sqrt{N_1}[X,X,X]_t\neq 0.$

Although individually each issue (a) or (b) has been studied in the literature, there is a lack of studies that take both the microstructure noise and time endogeneity effect into consideration.  \cite{RobertRosenbaum} study the estimation of the integrated (co-)volatility for an interesting model where the observation times are triggered by exiting from certain ``uncertainty
zones'', in which case both microstructure noise and time
endogeneity may exist. In this paper, we consider the presence of
both  microstructure noise and time endogeneity in a general
setting.

The paper is organized as follows. The  setup and assumptions are given in Section \ref{sec:setup}. The main results are given in Section \ref{sec:mainresult}. In Section \ref{sec:simulation}, simulation studies are performed in which our proposed estimator is compared with several existing popular estimators. Section \ref{sec:conclusion} concludes. The proofs (except that of Proposition \ref{Prop:bias_corr} below) are given in the Appendix; the proof of Proposition \ref{Prop:bias_corr} is given in the supplementary article \cite{LZZlama_supp}.

\section{Setup and assumptions}\label{sec:setup}
\begin{Ass}\label{asmtn:overall}
We assume the   setting of (\ref{eq:latentP}) and (\ref{eq:obsMSprice}) and that there is a filtration $(\mathcal{F}_t)_{t\geq0},$ with respect to which $W,$ $\mu$ and $\sigma$ in (\ref{eq:latentP}) are adapted and $(t_i)_{i\geq1}$ are $(\mathcal{F}_t)$-stopping times. Furthermore, the filtration $(\mathcal{F}_t)$ is generated by finitely many continuous martingales.
\end{Ass}
In the Introduction, we adopted the notation $N_1$ for the number of observed prices over time interval $[0,1].$ Here, we generalize this and denote
$$
N_t=\max\{i:t_i\leq t\}.
$$
In developing limiting results, one should be able to rely on some index variable approaching infinity/zero. In our context, we assume that $\max_{i} \Delta t_i\overset{p}\rightarrow0$ is driven by some underlying force, for instance, $n\rightarrow\infty,$ where $n$ (non-random) characterizes the  sampling frequency over time interval $[0,1].$

We aim at  effectively estimating  $\langle X,X\rangle_t$ based on our general setup.  A local averaging approach is adopted. We consider the time endogeneity  on the sub-grid level. Take the single sub-grid case for illustration,  the  sub-sample $\mathcal{S}=\mathcal{S}_0:=\{t_p,t_{p+q},\ldots,t_{p+iq},\ldots\}$ is constructed by choosing every $q$th observation (starting from the $p$th observation) from the complete grid. Here $p$ is the number of observations that we take in constructing local average,  $q$ is the size of blocks, and both are non-random numbers just as $n$.
Define
\begin{align}
\ell:=\lfloor\frac{n-p}{q}\rfloor,\label{eq:gpqrelationship}
\end{align}
which satisfies that $\ell q\leq n$, and as $p$ shall be taken as $o(n)$, $\ell q/n\to 1$ as $n\to\infty$.
As $n$ measures the sampling frequency of the complete grid,   $\ell$ measures that of the sub-grid~$\mathcal{S}$. Moreover, for notational ease, for $k=0,1,\ldots, q-1$, we define
\begin{equation}\label{dfn:t_ijk}
t_{i,j}^k:=t_{iq+p-j + k},\q \mbox{for } i=0,1,2,\ldots, \q\mbox{and }j=0,1,2,\ldots,p-1,
\end{equation}
and let
\[
t_{i,j}=t_{i,j}^0.
\]
Analogous to (\ref{eq:tricity}), we consider the quantity  $\sqrt{\ell}[X,X,X]_t^{\mathcal{S}}=\sqrt{\ell}\sum_{t_{i,0}\leq t}(X_{t_{i,0}}-X_{t_{i-1,0}})^3$.
The superscript $\mathcal{S}$ indicates the calculation being performed is based on the designated sub-grid. This convention applies to other sub-grids. Moving the sub-grid~$\mathcal{S}$ one step forward forms sub-grid $\mathcal{S}_1$, continuing this process gives sub-grid $\mathcal{S}_2$ and so on till the $(q-1)$th sub-grid $\mathcal{S}_{q-1}$. Figure \ref{fig:LAMA_grid_allo} provides a graphical demonstration of our grid allocation. Further, on the sub-grid $\mathcal{S}$, we define the number of observations up to time $t$ in sub-grid $\mathcal{S}$ as
$$
L_t:=\max\{i:t_{i,0}\leq t\}.
$$
Naturally, $L_1$ and $N_1$ satisfy $L_1\leq N_1/q.$
\begin{figure}[H]
    \centering
    \makebox{\includegraphics[width=6in, height=1.8in]
    {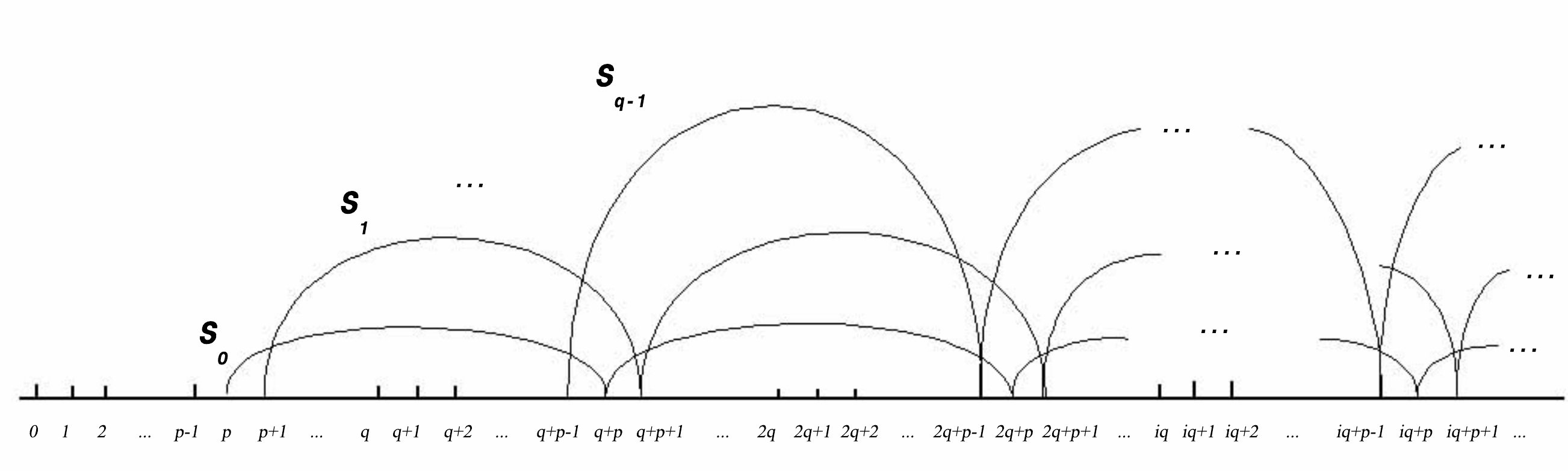}}
    \caption{\label{fig:LAMA_grid_allo}Grid allocation for Local Averaging.}
\end{figure}

\section{Main results}\label{sec:mainresult}
We start with results based on a single sub-grid and then proceed to the multiple sub-grids case.

\subsection{Single sub-grid: Local Averaging}

A natural and effective way of reducing the effect of microstructure noise in estimating $\langle X,X\rangle_t$ is averaging, see, e.g., \cite{JLMPV} and \cite{PV2009}. Following \citet{JLMPV}, we average every $p$ observations that precede each observation in the sub-sample $\mathcal{S}$ to obtain a new sequence of observations, which we denote by $(\overline{Y}_{t_{i,0}})_{i\geq 0}.$ Based on this sequence of observations, we obtain a single-grid biased local averaging estimator. To be specific,
\begin{align}
\overline{Y}_{t_{i,0}}=\frac{1}{p}\sum_{j=0}^{p-1}Y_{t_{i,j}},~{\rm for}~i=0,1,2,\ldots.\notag
\end{align}
The RV based on the $\overline{Y}$ sequence is denoted by
$$[\overline{Y},\overline{Y}]^{\mathcal{S}}_t=\sum_{t_{i,0}\leq t}(\Delta \overline{Y}_{t_{i,0}})^2,$$ where $\Delta \overline{Y}_{t_{i,0}}=\overline{Y}_{t_{i,0}}-\overline{Y}_{t_{i-1,0}}$ for $i\geq1$. After correcting the bias due to noise, the single-grid local averaging estimator is defined as
\begin{align}
\widehat{\langle X,X\rangle}_t^{LA}:=[\overline{Y},\overline{Y}]^{\mathcal{S}}_t-\frac{2L_t}{p}\widehat{\sigma_{\eps}^2},~{\rm for}~t\in[0,1],\label{eq:LAestimator}
\end{align}
where
\begin{align}
\widehat{\sigma_{\eps}^2}=[Y,Y]_1/2N_1,\label{eq:epsestimator}
\end{align}
is an estimator of $\sigma_\eps^2$, see Lemma \ref{Lem:ConvgcOfsigeps} in \ref{subsec:Prerequisitesandlemma}, and $[Y,Y]_1$ is the RV based on all observations up to time $1$.  We now state conditions that lead to the theorem for the single sub-grid case:
\begin{itemize}[noitemsep,nolistsep]
\item[C(1).] $\mu_t$ and $\sigma^2_t\geq c>0$ are  integrable and locally bounded;
\item[C(2).] $n/N_1=O_p(1)$;
\item[C(3).] $\Delta_n:=\max_{1\leq i\leq N_1}|t_i-t_{i-1}|=O_p(1/n^{1-\eta})$ for some nonnegative constant~$\eta$;
\item[C(4).] $L_t/\ell\overset{p}\longrightarrow \int_0^tr_sds~{\rm in}~D[0,1],$ where $r_s$ is an adapted integrable process (and hence in particular, $N_1/n=O_p(1)$);
\item[C(5).] The microstructure noise sequence $(\eps_{t_i})_{i\geq 0}$ consists of independent random variables with mean~0, variance $\sigma_{\eps}^2$, and common finite third and forth moments, and is independent of $\mathcal{F}_1$.
\end{itemize}
The following theorem characterizes the asymptotic property of the estimator (\ref{eq:LAestimator}).

\begin{Thm}\label{thm:single_grid}
Assume Assumption \ref{asmtn:overall} and conditions C(1)$\sim$ C(5). Suppose that $\eta\in[0,1/6)$, and $\ell\sim C_\ell n^{\alpha}$ and $p\sim C_pn^{\alpha}$ for some  $0<\alpha<2(1-\eta)/5$ and positive constants $C_\ell$ and $C_p$, and also that
\begin{align}
&\ell[X,X,X,X]_{t}^{\mathcal{S}}\overset{p}\longrightarrow\int_0^{t}u_s\sigma_s^4\ ds~ {\rm for~every~}t\in[0,1],\quad ~{\rm and} \label{quarticity}\\
&\sqrt{\ell}[X,X,X]_{t}^{\mathcal{S}}\overset{p}\longrightarrow\int_0^{t}v_s\sigma_s^3\ ds~{\rm for~every~}t\in[0,1],\label{tricity}
\end{align}
where $[X,X,X,X]_{t}^{\mathcal{S}}=\sum_{t_{i,0}\leq t}(X_{t_{i,0}}-X_{t_{i-1,0}})^4$, and $u_s\sigma_s^4$ and $v_s^2\sigma_s^4$ are both integrable. Then, stably in law,
\begin{align}
\sqrt{\ell}\left(\widehat{\langle X,X\rangle}_t^{LA}-\langle X,X\rangle_t\right)&\Longrightarrow \underbrace{\frac{2}{3}\int_0^tv_s\sigma_sdX_s}_{asymptotic~bias}+\notag\\
&\int_0^t\left[\left(\frac{2}{3}u_s-\frac{4v_s^2}{9}\right)\sigma_s^4+12r_s\left(\frac{C_\ell}{C_p}\sigma_\eps^2\right)^2+8\frac{C_\ell}{C_p}\sigma_s^2\sigma_\eps^2\right]^{1/2}dB_s\notag
\end{align}
where $B_t$ is a standard Brownian motion independent of
$\mathcal{F}_1$.
\end{Thm}
Proof of the theorem is given in \ref{ssec:pf_single_grid}.

In the literature it is often assumed that the mesh $\Delta_n=O_p(1/n)$, in other words, $\eta=0$ in Condition C(4). In this case,
the convergence rate in Theorem \ref{thm:single_grid} can be arbitrarily close to $n^{1/5}.$

\begin{Rem}
\label{Rem:Sub-grid Endogeneity}
Unlike in the full grid setting where a nonzero limit of tricity can be easily generated by letting the sampling times be hitting times of asymmetric barriers (see for instance Examples 4 \& 5 of \cite{lmrzz09}), in the subgrid case a nonzero limit of tricity is far less common, and in particular under the settings of both Examples 4 \& 5 of \cite{lmrzz09}, the limit in $\eqref{tricity}$ vanishes. However as we found in simulation studies (not all reported), even in these situations, adopting the (finite sample) bias correction discussed in Section \ref{ssec:bias_cor} below can substantially reduce the (finite sample) bias. Similar remark applies to the estimator in Theorem \ref{thm:multiple_grids} below.
\end{Rem}

\subsection{Multiple sub-grids: Moving Average}
We show in this subsection that for any $\eps>0$, rate $n^{1/4-\eps}$ consistency can be achieved by using moving average based on multiple sub-grids. For that purpose, we need such notations as $[\overline{Y},\overline{Y}]^{\mathcal{S}_k}_t$, i.e. the RV of locally averaged $Y$ process over the $k$th sub-grid, for the same operations that are performed over the $0$th sub-grid $\mathcal{S}=\mathcal{S}_0$ being adjusted to the $k$th sub-grid $\mathcal{S}_k$. To be specific, we take $[\overline{Y},\overline{Y}]^{\mathcal{S}_k}_t$ for example; other notation with superscript $k$ or $\mathcal{S}_k$ has similar interpretation. Similar to the definition of $[\overline{Y},\overline{Y}]^{\mathcal{S}}_t$ (i.e. $[\overline{Y},\overline{Y}]^{\mathcal{S}_0}_t$), we first define
$$
\overline{Y}_{t_{i,0}^k}:=\frac{1}{p}\sum_{j=0}^{p-1}Y_{t_{iq+p+k-j}},~{\rm for}~i=0,1,2,\ldots,
$$
where, recall that $t_{i,0}^k=t_{iq+p+k}$ denotes the $i$th observation time on the $k$th sub-grid. The RV of locally averaged $Y$ process over the $k$th sub-grid is defined as follows
$$
[\overline{Y},\overline{Y}]^{\mathcal{S}_k}_t:=\sum_{t_{i,0}^k\leq t}(\Delta\overline{Y}_{t_{i,0}^k})^2,
$$
where $\Delta\overline{Y}_{t_{i,0}^k}=\overline{Y}_{t_{i,0}^k}-\overline{Y}_{t_{i-1,0}^k}$ for $i\geq1$. Assume the following conditions that lead to the asymptotic result on multiple sub-grids:
\begin{itemize}[noitemsep,nolistsep]
\item[C(6).] $\ell\sum_{t_{i}\leq t}\left(\sum_{j=1}^{q-1}\frac{q-j}{q}\Delta X_{t_{i-j}}\right)^2(\Delta X_{t_i})^2\overset{p}\longrightarrow\int_0^{t}w_s\sigma_s^4\ ds$ for every $t\in[0,1]$, where $w_s\sigma_s^4$ is integrable;
\item[C(7).] $\frac{1}{q}\sum_{k=0}^{q-1}\sqrt{\ell}[X,X,X]_{t}^{\mathcal{S}_k}\overset{p}\longrightarrow\int_0^{t}\bar{v}_s\sigma_s^3\ ds$ for every $t\in[0,1]$, where $\bar{v}_s^2\sigma_s^4$ is integrable.
\end{itemize}
Define
$$
A(p,q):=\frac{2}{q}\sum_{j=1}^{p-1}\left(\frac{j^2}{p^2}-\frac{j}{p}\right).
$$
Under the conditions of Theorem \ref{thm:multiple_grids} below, $A(p,q)\sim -n^{4\alpha-2}C_\ell C_p/3.$
\begin{Thm}\label{thm:multiple_grids}
Assume Assumption 1 and conditions C(1) to C(7). Suppose that $\eta\in[0,1/9)$, and $\ell\sim C_\ell n^{\alpha}$ and $p\sim C_pn^{3\alpha-1}$ for some $\max(4\eta,1/3)<\alpha<(1-\eta)/2$ and positive constants $C_\ell$ and $C_p$. Then, stably in law,
\begin{align}
&~~~\sqrt{\ell}\left(\frac{1}{q}\sum_{k=0}^{q-1}[\overline{Y},\overline{Y}]^{\mathcal{S}_k}_t-\frac{2N_t}{pq}\widehat{\sigma_{\eps}^2}-(1+A(p,q))\langle X,X\rangle_t\right)\notag\\
&\Longrightarrow \underbrace{\frac{2}{3}\int_0^t\bar{v}_s\sigma_s\ dX_s}_{asymptotic~bias}+\int_0^t\left[\left(4w_s-\frac{4}{9}\bar{v}_s^2\right)\sigma_s^4+\frac{8C_\ell^3}{C_p}r_s(\sigma_\eps^2)^2\right]^{1/2}~dB_s\notag,
\end{align}
where $B_t$ is a standard Brownian motion that is independent of $\mathcal{F}_1$.
\end{Thm}
Proof of the theorem is given in \ref{ssec:pf_multiple_grid}.

If one assumes that $\Delta_n=O_p(1/n)$, then $\eta=0$, and the convergence rate in the above theorem can be arbitrarily close to $n^{1/4}.$

\begin{Rem}
If times are exogenous, Condition C(6) can be reduced to a similar assumption as (48) on p.1401 of
\cite{ZhangMyklandAitSahalia}. The limit is then related to quarticity and
can be consistently estimated, see, e.g., \cite{JLMPV}, \cite{NHLS}. In general, when observation times can be endogenous,
the limit is expected to be different.
\end{Rem}

\subsection{Bias Correction}\label{ssec:bias_cor}
Since the estimator constructed based on multiple grids achieves a better rate of convergence, below we shall mainly focus on the moving average setting. Based on the above result, we have the following (infeasible) unbiased  estimator:
$$
\widehat{\langle X,X\rangle}^{(0)}_1:=\frac{-\frac{2}{3}\int_0^1\bar{v}_s\sigma_sdX_s+\sqrt{\ell}\left(\frac{1}{q}\sum_{k=0}^{q-1}[\overline{Y},\overline{Y}]^{\mathcal{S}_k}_1-\frac{2N_1}{pq}\widehat{\sigma_{\eps}^2}\right)}{\sqrt{\ell}(1+A(p,q))}.
$$
The following Corollary describes the asymptotic property for this  estimator.
\begin{Cor}
Under the assumptions of Theorem 2, stably in law,
\begin{align}
\sqrt{\ell}\left(\widehat{\langle X,X\rangle}^{(0)}_1-\langle X,X\rangle_1\right)\Longrightarrow\int_0^1\left[\left(4w_s-\frac{4}{9}\bar{v}_s^2\right)\sigma_s^4+\frac{8C_\ell^3}{C_p}r_s(\sigma_\eps^2)^2\right]^{1/2}~dB_s,\notag
\end{align}
where $B_t$ is a standard Brownian motion that is independent of
$\mathcal{F}_1$.
\label{Cor:1}
\end{Cor}
\begin{proof} This is just a rearrangement of the convergence in Theorem 2.\end{proof}

To improve over $\widehat{\langle X,X\rangle}^{(0)}_1$ and build a feasible unbiased  estimator, a consistent estimator for the bias term $2/3\int_0^t\bar{v}_s\sigma_s\ dX_s$ is needed. This is the issue that we deal with next. Define
\begin{align} F^{(2)}_n(t)&:=\frac{\sqrt{\ell}\left(\frac{1}{q}\sum_{k=0}^{q-1}[\overline{Y},\overline{Y}]^{\mathcal{S}_k}_t
-\frac{2N_t}{pq}\widehat{\sigma_{\eps}^2}\right)}{\sqrt{\ell}(1+A(p,q))},\q\mbox{and}\q f^{(2)}(t):=\sigma_t^2, \label{uncorrected_est}\\
 F^{(3)}_n(t)&:=\frac{1}{q}\sum_{k=0}^{q-1}\sqrt{\ell}[\overline{Y},\overline{Y},\overline{Y}]^{\mathcal{S}_k}_{t},
 \q\mbox{and}\q f^{(3)}(t):=\bar{v}_t\sigma_t^3. \notag
\end{align}
For a given partition $(\tau_i)_{i\geq0}$ over $[0,1]$, we define
\begin{equation}\label{dfn:derivative}
f^{(j)}_n(t):=(F^{(j)}_n(\tau_i)-F^{(j)}_n(\tau_{i-1}))/(\tau_i-\tau_{i-1}),~{\rm for}~t\in[\tau_i,\tau_{i+1}), ~{\rm for}~j=2,3.
\end{equation}
We then have that stably in law,
$$
\sqrt{\ell}\left(F^{(2)}_n(t)-\langle X,X\rangle_t\right)\Longrightarrow \frac{2}{3}\int_0^t\bar{v}_s\sigma_sdX_s+\int_0^t\left[\left(4w_s-\frac{4}{9}\bar{v}_s^2\right)\sigma_s^4+\frac{8C_\ell^3}{C_p}r_s(\sigma_\eps^2)^2\right]^{1/2}~dB_s.
$$
Define
$$
\gamma(\alpha,\eta):=\min\{-2\alpha+1-3\eta/2;\alpha/2-\eta/2;7\alpha/2-3/2;3\alpha/2-1/2-\eta;5\alpha/2-1-\eta/2\}.
$$
And assume
\begin{itemize}[noitemsep,nolistsep]
\item[C(7')] $\left|\frac{1}{q}\sum_{k=0}^{q-1}\sqrt{\ell}[X,X,X]_{t}^{\mathcal{S}_k}-\int_0^{t}\bar{v}_s\sigma_s^3ds\right|/\delta_{n}\overset{p}\longrightarrow0$ in $D[0,1]$ for a (nonrandom) sequence $(\delta_{n})_{n\geq1}$ with $\delta_{n}\rightarrow0$ and $1/\delta_{n}=o\left(n^{\gamma(\alpha,\eta)}\right)$.
\end{itemize}
We have the following
\begin{Pro}\label{Prop:bias_corr}
Assume the conditions of Theorem 2, C(7') and $3/7<\alpha<(2-3\eta)/4$ with $\eta\in[0,2/21).$ Suppose $f^{(j)}(t)$ is a.s. continuous and bounded on $[0,1]$ for $j=2,3$. Moreover, define a partition $[\tau_i,\tau_{i+1}]:=[t_{id_1q},t_{(i+1)d_1q}]$ which is a block of $d_1q$ time intervals over the complete grid with $1/d_1=o\left(1/n^{1-2\alpha}\right)$, $\max_i|\tau_i-\tau_{i-1}|=o_p(1)$ and $\delta_n/\min_i|\tau_i-\tau_{i-1}|=O_p(1)$; and let
$$
\Delta\overline{Y}_{\tau_i}:=\frac{1}{p}\sum_{j=0}^{p-1}Y_{t_{id_1q+p-j}}-\frac{1}{p}\sum_{j=0}^{p-1}Y_{t_{(i-1)d_1q+p-j}}.
$$
Then
$$
\sum_{\tau_{i}\leq t}\frac{f^{(3)}_n(\tau_{i-1})}{f^{(2)}_n(\tau_{i-1})}\Delta\overline{Y}_{\tau_i}\overset{p}\longrightarrow
\int_0^t\bar{v}_s\sigma_s\ dX_s ~in~D[0,1].
$$
\end{Pro}
Proof of Proposition \ref{Prop:bias_corr} is given
in the supplementary article \cite{LZZlama_supp}.

According to the above proposition, a consistent estimator for the bias $2/3\int_0^1\bar{v}_s\sigma_s\ dX_s$ is given by
$$
\mathcal{B}:=\frac{2}{3}\sum_{\tau_{i}\leq 1}\frac{f^{(3)}_n(\tau_{i-1})}{f^{(2)}_n(\tau_{i-1})}\Delta\overline{Y}_{\tau_i}.
$$
Finally, we define our feasible unbiased estimator as
$$
\widehat{\langle X,X\rangle}_1:=\frac{-\mathcal{B}+\sqrt{\ell}\left(\frac{1}{q}\sum_{k=0}^{q-1}[\overline{Y},\overline{Y}]^{\mathcal{S}_k}_1-\frac{2N_1}{pq}\widehat{\sigma_{\eps}^2}\right)}{\sqrt{\ell}(1+A(p,q))}.
$$
The following theorem gives the CLT for our final estimator.
\begin{Thm}
Under the assumptions of Theorem 2 and Proposition 1, stably in law,
\begin{align}
\sqrt{\ell}\left(\widehat{\langle X,X\rangle}_1-\langle X,X\rangle_1\right)\Longrightarrow \int_0^1\left[\left(4w_s-\frac{4}{9}\bar{v}_s^2\right)\sigma_s^4+\frac{8C_\ell^3}{C_p}r_s(\sigma_\eps^2)^2\right]^{1/2}~dB_s,\notag
\end{align}
where $B$ is a standard Brownian motion independent of $\mathcal{F}_1$.
\end{Thm}

\section{Simulation studies}\label{sec:simulation}
\setcounter{equation}{0}
\renewcommand{\theequation}{\thesection.\arabic{equation}}

In this section, we conduct simulation studies. We investigate the performance of our proposed estimator $\widehat{\langle X,X\rangle}_1$ compared with  existing popular estimators in both endogenous and non-endogenous cases. We shall use two data generating mechanisms for $X$: (1) a constant volatility Brownian bridge; and (2) a stochastic volatility Heston bridge.  In each case, we start the latent process $X$ at $X_0=\log(5)$, let the standard deviation of the noise be $\sigma_{\eps}:=(\sigma_\eps^2)^{1/2}=0.0005$ and simulate 1,000 sample paths for observed price process $Y$.

\subsection{Estimators used for comparison}\label{ssec:4_est}
Below we briefly recall four commonly used volatility estimators: the  two scales realized volatility (TSRV) of \citet{ZhangMyklandAitSahalia}, the  multi-scale realized volatility (MSRV) of \citet{Zhang}, the Realized Kernel estimator of \citet{NHLS},  and the Pre-averaging estimator of \citet{JLMPV}.

The (small-sample adjusted) TSRV estimator is given by
$$
\widehat{\langle X,X\rangle}_1^{tsrv}=\left(1-\frac{1}{K_{tsrv}}\right)^{-1}\left(\frac{1}{K_{tsrv}}\sum_{k=1}^{K_{tsrv}}[Y,Y]_1^{(k)}-\frac{1}{K_{tsrv}}[Y,Y]_1\right),
$$
where the data is divided into $K_{tsrv}$ non-overlapping sub-grids and $[Y,Y]_1^{(k)}$ is the RV on the $k$th sub-grid. \citet{ZhangMyklandAitSahalia} provided a guideline on the choice of the grid allocation. If we pretend that the volatility were constant, then  the optimal choice for grid allocation is $K_{tsrv}=c_{tsrv}N_1^{2/3},$ where, in practice, one can set
$
c_{tsrv}=\left(\frac{12([Y,Y]_1/(2N_1))^2}{([Y,Y]_1^{sub})^2}\right)^{1/3},
$
where $[Y,Y]_1^{sub}$ is the RV based on sparse sampling. Here, we implement $[Y,Y]_1^{sub}$ at 5 minutes frequency.

The MSRV estimator, which is a rate-optimal extension to  TSRV, is given as follows
$$
\widehat{\langle X,X\rangle}_1^{msrv}=\sum_{j=1}^{K_{msrv}}\lambda_{j}\frac{1}{j}\sum_{k=1}^{j}[Y,Y]_1^{(k)},
$$
where $\lambda_1=a_1+((N_1+1)/2)^{-1}$, $\lambda_2=a_2-((N_1+1)/2)^{-1}$ and $\lambda_i=a_i$ for $i\geq3$
with $a_i=h(i/K_{msrv})i/K_{msrv}^2-h'(i/K_{msrv})i/(2K_{msrv}^3),$ for $i=1,\ldots,K_{msrv},$ where $K_{msrv}=c_{msrv}N_1^{1/2}$ and $h(x)=12x-6$.
The optimal choice of $c_{msrv}$ when the volatility is constant is
$$
c_{msrv}=\left(\frac{T_3+T_4+\left((T_3+T_4)^2+12T_1T_2\right)^{1/2}}{2T_2}\right)^{1/2},
$$
where $T_1=48([Y,Y]_1/(2N_1))^2$, $T_2=52([Y,Y]_1/(2N_1))^2/35$, $T_3=24([Y,Y]_1/(2N_1))^2/5$ and $T_4=48[Y,Y]^{sub}_1([Y,Y]_1/(2N_1))/5$.

The Realized Kernel estimator is defined as
$$
\widehat{\langle X,X\rangle}_1^{Ker}=[Y,Y]_1+\sum_{h=1}^Hf_{k}((h-1)/H)\Big[\sum_{i=1}^{N_1}\left(\Delta Y_{t_i}\Delta Y_{t_{i-h}}+\Delta Y_{t_i}\Delta Y_{t_{i+h}}\right)\Big],
$$
where $H=c_{ker}N_1^{1/2}$ and $f_{k}$ is a kernel function. We  choose   the Parzen kernel:
$$
f_k(x)=\left\{\begin{array}{lll}
1-6x^2+6x^3~~~{\rm for}~0\leq x\leq1/2;\\
2(1-x)^3~~~{\rm for}~1/2\leq x\leq1.
\end{array}
\right.
$$
Under constant volatility, the optimal choice for $c_{ker}$ in practice is given by
$$
c_{ker}=\left(\frac{[Y,Y]_1/(2N_1)}{[Y,Y]^{sub}_1}\right)^{1/2}\left(\frac{1}{f_k^{0,0}}\left(-f_k^{0,2}+\left((f_k^{0,2})^2+3f_k^{0,0}(f_k'''(0)+f_k^{0,4})\right)^{1/2}\right)\right)^{1/2},
$$
where $f_k^{0,0}=\int_0^1f_k(x)^2dx$, $f_k^{0,2}=\int_0^1f_k(x)f_k''(x)dx$ and $f_k^{0,4}=\int_0^1f_k(x)f_k'''(x)dx$.

The Pre-averaging estimator  is  as follows:
$$
\widehat{\langle X,X\rangle}_1^{Pre}=\frac{1}{\theta\varphi_2\sqrt{N_1}}\sum_{i=0}^{N_1-k_n+1}(\Delta\widehat{Y}_i)^2-\frac{\varphi_1}{2\theta^2\varphi_2N_1}[Y,Y]_1,
$$
where $\varphi_1=1,$ $\varphi_2=1/12$ and
$$
\Delta\widehat{Y}_i=\frac{1}{k_n}\left(\sum_{j=k_n/2}^{k_n-1}Y_{i+j}-\sum_{j=0}^{k_n/2-1}Y_{i+j}\right)
$$
with $k_n=\sqrt{N_1}\theta.$
The optimal choice of $\theta$   when the volatility is constant is
\[
\theta=4.777([Y,Y]_1/(2N_1))^{1/2}/([Y,Y]^{sub}_1)^{1/2}.
\]

\begin{Rem}
The grid allocation schemes in constructing the above estimators are optimal in the sense of achieving efficient asymptotic variance bound when $(\sigma_t)$ is constant. However, in practice there is   no optimal choice since, for instance, $(\sigma_t)$ is random and time dependent. See Remarks 2 and 3 in \citet{JLMPV} for related discussions on this. In our case, due to the more complex model assumptions, i.e. data with time endogeneity and noise, and grid allocation scheme, i.e. bivariate setting $(p,q)$ in contrast to the existing univariate cases, we do not provide a theoretical optimal  choice but rather give below some practical guidelines.\label{rem:gridalloremark}
\end{Rem}

Back to our estimator $\widehat{\langle X,X\rangle}_1$, there are several tuning parameters ($n, \ell, p, q$ and $d_1$) that one has to determine. Regarding $n$ which characterizes the sampling frequency, one can use the average number of transactions per day for the past, say 30, days as an approximation.  About $(\ell, p,q)$,
notice that Theorem \ref{thm:multiple_grids} suggests $\ell\sim C_\ell n^{\alpha}$ (hence $q\sim n/\ell$)
and $p\sim C_pn^{3\alpha-1}$.
On the one hand, one should choose $\ell$ as large as possible in order to have higher convergence rate. On the other hand, large $\ell$ induces small $q$ and hence small $p$ (recall $q>p$) and the main role that $p$ plays is to reduce the microstructure noise. Hence, one should also be aware of the magnitude of the microstructure noise when choosing appropriate $p$, and $p$ can not be too small when prices are heavily contaminated. Under the simulation setting below, the sampling frequency is around $n=46,800$, and the standard deviation of the noise is $\sigma_{\eps}=0.0005$. We choose $p=5$ which is found to be good enough to reduce the microstructure noise effect. In practice, one can use \eqref{eq:epsestimator} to estimate the standard deviation of the noise and come up with a reasonable choice of $p$. The block size $q$ should be larger than $p$ and is chosen as 20 (and $\ell\approx 2,340$). As to $d_1$, this depends on, for example, how volatile the volatility process is, which one can get some rough idea by looking at a suitable estimate of the spot volatilities. If the volatility process is more volatile, one should divide the whole time interval into shorter time periods, i.e., choose a smaller $d_1$. In our simulation, we choose   $d_1=100$, i.e. dividing the complete grid into around 20 blocks.

We next present our three simulation designs and the corresponding results.

\subsection{Design I: Brownian bridge with hitting times}
We first consider the case when the latent price process $X$ follows a Brownian bridge with (constant) volatility $\sigma$ that starts at $X_0$ and ends at $X_0+4\sigma$. $X$ can be expressed as (see pp.358 of \citet{KaratzasShreve})
$$
dX_t=\frac{X_0+4\sigma-X_t}{1-t}\ dt+\sigma\ dW_t,
$$
where $W_t$ is a standard Brownian motion. In this study, we set $\sigma=0.02.$
The sampling times are generated as follows: let $a=5\sigma$, $b=\sigma/10$, $n=46,800$, $\ell'\approx 16800$ (roughly $n^{19/21}$), and $q'=[n/\ell']$.  Then
\begin{itemize}[noitemsep,nolistsep]
\item[(1)] For $j=0,1,2,\ldots,q',$ $t_j=\frac{j}{2n};$
\item[(2)] For $i=1,2,\ldots$,
    \begin{itemize}
    \item[] Sparse sampling: $t_{iq'+1}=\inf\{t>t_{iq'}:X_t-X_{t_{iq'}}= ~{\rm either}~a/\sqrt{\ell'}~{\rm or~}-b/\sqrt{\ell'}\}$;
    \item[] Intensive sampling: $t_{iq'+j}=t_{iq'+1}+\frac{j-1}{2n},$ for $j=2,\ldots,q'$.
    \end{itemize}
\end{itemize}
The mean observation duration when sampling sparsely is about $1/(2\ell')$, roughly 3 times of the observation duration when sampling intensively. If as $n\to\infty$, $\ell'$ grows in the rate of $n^{19/21}$, then actually the limit in C(7) vanishes, however, as one can see from the simulation results below, (finite sample) bias correction as discussed in Subsection \ref{ssec:bias_cor} can substantially reduce the (finite sample) bias.

Figure \ref{Fig:hist_qq_design_I} displays the histogram and normal Q-Q plot for the estimator $\widehat{\langle X,X\rangle}_1$ based on the 1,000 simulated samples. The plots show that the finite sample behavior of our CLT works well. In Table \ref{tab:com_designI} we compare the performances of the four estimators that we discussed in Section \ref{ssec:4_est}, the ``Uncorrected'' estimator $F_n^{(2)}(1)$ defined in \eqref{uncorrected_est}, and our final estimator $\widehat{\langle X,X\rangle}_1$. From the table one can see that our estimator provides the smallest RMSE and has substantially smaller bias than the others (reduced by more than 80\%) while maintains similar efficiency (standard deviation).

\begin{figure}[H]
\begin{center}
\includegraphics[angle=90,width=7cm]{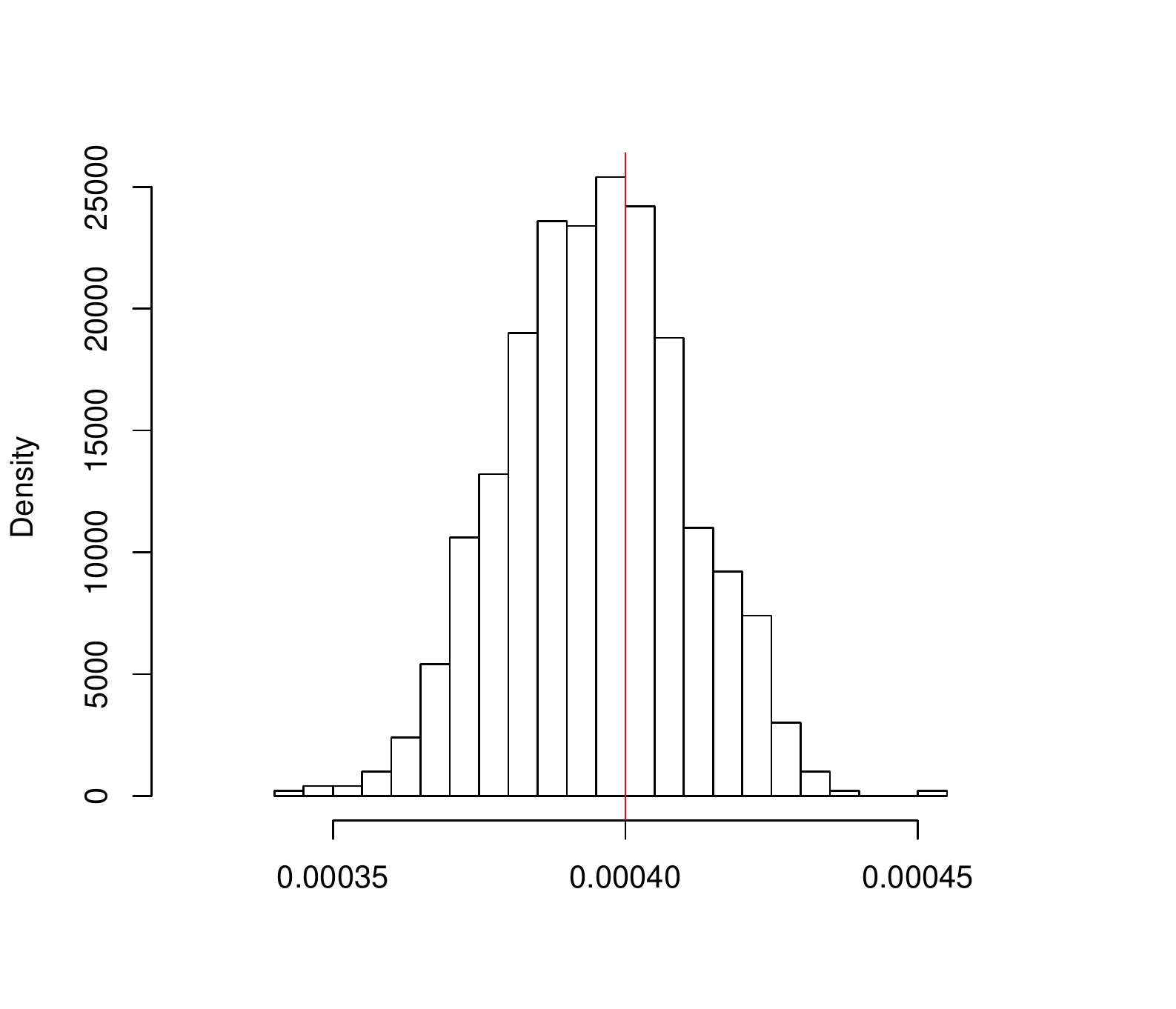}
\includegraphics[angle=90, width=7cm]{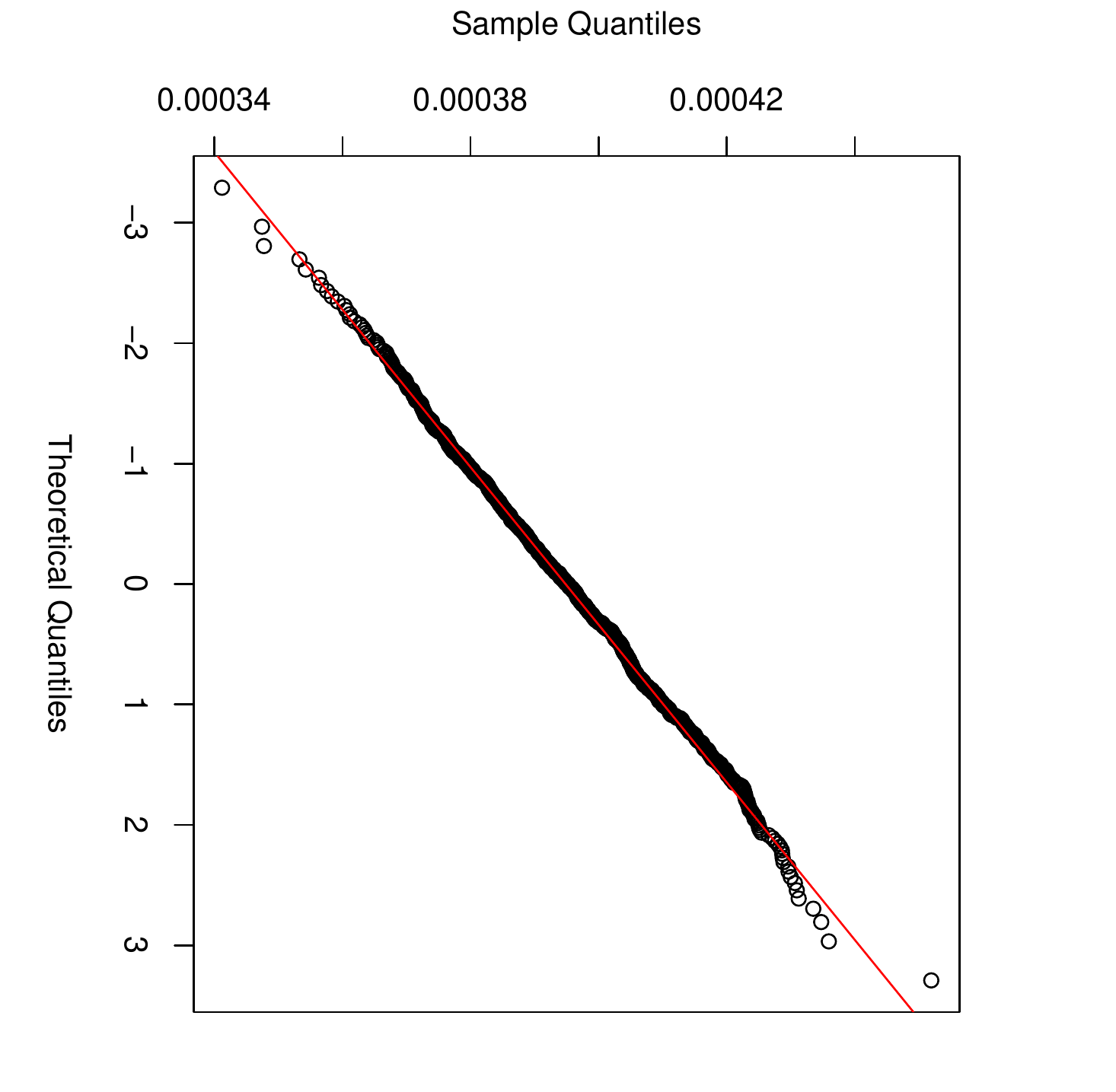}
    \caption{\label{Fig:hist_qq_design_I}Histogram and QQ plot of the estimator $\widehat{\langle X,X\rangle}_1$ for Design I. The red vertical line in the histogram indicates the true value of target.}
\end{center}
\end{figure}

\begin{table}[H]
   {\small \caption{\label{tab:com_designI}Performance of the six estimators in the presence of endogenous time for Design I, the constant volatility case. Our estimator $\widehat{\langle X,X\rangle}_1$ provides the smallest RMSE.  The RMSE is reduced by more than 50\%;  the bias is reduced by more than 80\% while
     the standard deviation is kept at the same level as others. }
\fbox{%
\begin{tabular}{*{7}{c}}
    &    TSRV     &MSRV           & Kernel & Pre-averaging & Uncorrected&$\widehat{\langle X,X\rangle}_1$ \\
\hline
RMSE   &3.734e-05 & 3.553e-05  &3.810e-05  & 3.340e-05  &3.300e-05 &1.621e-05 \\
sample bias & 3.300e-05& 3.163e-05 &  3.454e-05& 2.927e-05& 2.911e-05& -4.997e-06\\
sample s.d. & 1.748e-05 & 1.619e-05  &  1.609e-05& 1.609e-05&1.555e-05 &1.543e-05   \\
\end{tabular}}}
\end{table}

\subsection{Design II: Heston Bridge with hitting times}
In order to further investigate the performance of our estimator under more complex situations, in this subsection, we consider the following stochastic volatility model
\[
\left\{
\aligned
&dX_t=\frac{X_0+4\vartheta^{1/2}-X_t}{1-t}\ dt+\sqrt{V_t}\ dW_t\notag\\
&dV_t=\kappa(\vartheta-V_t)\ dt+\gamma\sqrt{V_t}\ dW^\sigma_t\notag,
\endaligned
\right.
\]
where $W_t$ and $W_t^\sigma$ are standard Brownian motions with instantaneous correlation coefficient $\rho$, and $\kappa,$ $\vartheta$ and $\gamma$ are positive constants. We consider the situation when $X$ starts at $X_0$ and ends at $X_0+4\vartheta^{1/2}$. In the simulation, we set $\vartheta=0.0004$, $\gamma=0.5/252,$ $\kappa=5/252$ and $\rho=-0.5.$ Here, we choose a moderate value $-0.5$ for $\rho$ to represent the leverage effect.
The leverage effect can be bigger for indices as studied by
\cite{YacineKimmel} and \cite{AFL}.  Times are generated according to the same hitting rule as in Design I. We can see from Table \ref{tab:com_designII} that in this more complex situation, our estimator again has substantially smaller bias and RMSE than the others. We did not include the sample standard deviation here since the integrated volatility to be estimated in this case depends on the sample path and is random.
\begin{table}[H]
    {\small\caption{\label{tab:com_designII}Performance of the six estimators in the presence of endogenous time for Design II,
    the stochastic volatility case. Our estimator again provides the smallest RMSE. The RMSE is reduced by more than 50\%;
    the bias is reduced by more than 80\%.
   }
\fbox{%
\begin{tabular}{*{7}{c}}
    &    TSRV    &MSRV           & Kernel  & Pre-averaging & Uncorrected &$\widehat{\langle X,X\rangle}_1$ \\
\hline
RMSE   & 3.824e-05&  3.579e-05 &3.835e-05   &3.387e-05 & 3.375e-05&1.636e-05\\
sample bias &3.393e-05   & 3.175e-05  &3.463e-05   & 2.965e-05 & 2.974e-05 &-4.215e-06 \\
\end{tabular}}}
\end{table}

\subsection{Design III: Brownian Bridge with independent Poisson times}
The goal of this design is to check the performance of our estimator when the sampling times are not endogenous. We again assume the Brownian bridge dynamic for $X$ as in Design I. The observation times are now generated from an independent Poisson process with rate 46,800.   Table~\ref{tab:com_designIII} reports the result of performance comparison, and we can see that our estimator performs similarly as the other estimators in this case.

\begin{table}[H]
    {\small\caption{\label{tab:com_designIII}Performance of the six estimators
     when the observation times are \textbf{\emph{not}} endogenous. The performance of
    our estimator  is comparable to others. }
\fbox{%
\begin{tabular}{*{7}{c}}
    &    TSRV     &MSRV          &  Kernel& Pre-averaging & Uncorrected&$\widehat{\langle X,X\rangle}_1$\\
\hline
RMSE   & 1.486e-05& 1.375e-05 & 1.434e-05 & 1.373e-05 &1.312e-05&1.568e-05\\
sample bias &  2.643e-06&1.584e-06 & 4.144e-06 &-2.847e-07  &-1.274e-06&-7.723e-06 \\
sample s.d. &  1.463e-05&1.367e-05 & 1.374e-05 &1.373e-05 &1.307e-05&1.365e-05 \\
\end{tabular}}}
\end{table}

In summary, one observes from Tables 1-3 that when sampling times are endogenous (Designs I and II),
one can have substantial reductions in RMSE and bias by using our estimator. When there is no endogeneity (Design III), our estimator performs comparably to others.

\section{Concluding remarks}\label{sec:conclusion}
In this paper, we establish a theoretical framework for dealing with effects of both the endogenous time and microstructure noise in volatility inference. An estimator that can accommodate both issues is proposed. Numerical studies are performed. The results show that our proposed estimator can substantially outperform existing popular estimators when time endogeneity exists, while has a comparable performance to others when there is no endogeneity.

\section*{Acknowledgements}
We thank Rainer Dahlhaus, Jean Jacod, Per Mykland and Nakahiro Yoshida for this special issue, and we are very grateful to an anonymous referee for very careful reading of the
paper and constructive suggestions.

\appendix

\section{Proofs}\label{sec:proofs}
Throughout the proofs, $C, c, C_1,$ \hbox{etc.} denote generic constants whose
values may change from line to line.
Moreover, since we shall establish stable convergence, by a change of measure argument (see e.g. Proposition 1 of \cite{MZ2012}) we can suppress the drift and assume that
\begin{itemize}[noitemsep,nolistsep]
\item[1.] $\mu_t\equiv0$.
\end{itemize}
Moreover, because of the local boundedness condition on  $\sigma_t^2$, by standard  localization arguments we can assume without loss of generality that
\begin{itemize}[noitemsep,nolistsep]
\item[2.] $0<c\leq \sigma_t\leq \sigma_+,$ where $c$ and $\sigma_+$ are nonrandom numbers,
\end{itemize}
see e.g. \citet{MZ2009b} and \citet{MZ2012}.
Similarly, we can without loss of generality strengthen the assumption on $\Delta_n$ and $N_1$ in C(2) -- C(4) as follows:
\begin{itemize}[noitemsep,nolistsep]
\item[3.] $\Delta_n\leq C/n^{1-\eta}$; and
\item[4.] $n/C\leq N_1\leq Cn.$
\end{itemize}

\subsection{Prerequisites}\label{subsec:Prerequisitesandlemma}
In the proofs, we shall repeatedly use the following inequalities.

\noindent{\it Burholder-Davis-Gundy (BDG) inequality with random times:}\\
First, if $t_i$'s are stopping times and $f(s)$ is adapted with $\max_{0\leq s\leq1}|f(s)|\leq f_+$, then by the Burholder-Davis-Gundy inequality with random times (see, e.g., p. 161 of \citet{RevuzandYor1999}), for any exponent $\beta\geq1$,
\[
E\left(\int_{t_{i-1}}^{t_i}f(s)dW_s\right)^\beta
\leq C E\left(\int_{t_{i-1}}^{t_i}f(s)^2ds\right)^{\beta/2}\notag
\leq Cf_+^\beta E(t_i-t_{i-1})^{\beta/2}.
\]

\noindent{\it Doob's $L^p$ inequality}:\\
Second, for any process $Z$, which is either a continuous time martingale or  a positive submartingale, Doob's $L^p$ inequality (see p.54 of \citet{RevuzandYor1999}) states that, for any $\beta\geq1$ and any $\lambda>0,$
$$
P\Big[\sup_{s\in[0,1]}|Z_s|\geq \lambda\Big]\frac{1}{\lambda^\beta}E|Z_1|^\beta,
$$
and for $\beta>1$,
$$
\left(E\Big[\sup_{s\in[0,1]}|Z_s|\Big]^\beta\right)^{1/\beta}
\leq \frac{\beta}{\beta-1}\left(E|Z_1|^\beta\right)^{1/\beta}.
$$
Therefore, if we can establish a bound order for $E|Z_1|^\beta$ ($\beta=1~{\rm or}~2$ in our case), then the same bound order applies in $D[0,1]$.

We will also use the following results about the
convergence of $\widehat{\sigma^2_\eps}$ to $\sigma_\eps^2$.
\begin{Lem}
For $\widehat{\sigma_{\eps}^2}$ defined in (\ref{eq:epsestimator}), one has $\sqrt{N_1}\left(\widehat{\sigma_{\eps}^2}-\sigma_\eps^2\right)=O_p(1)$.
\label{Lem:ConvgcOfsigeps}
\end{Lem}

\begin{proof} First, notice that $$\sqrt{N_1}\left(\widehat{\sigma_{\eps}^2}-\sigma_\eps^2\right)=[X,X]_1/2\sqrt{N_1}+[X,\eps]_1/\sqrt{N_1}+([\eps,\eps]_1-2N_1\sigma_\eps^2)/2\sqrt{N_1}.$$ By C(2) and the fact that $[X,X]_1=O_p(1)$,
$$[X,X]_1/2\sqrt{N_1}=O_p\left(\frac{1}{\sqrt{n}}\right).$$
As to $[X,\eps]_1/\sqrt{N_1},$ we treat it as follows,
\begin{align}
[X,\eps]_1/\sqrt{N_1}=
\frac{1}{\sqrt{N_1}}\sum_{t_{i}\leq 1}\Delta X_{t_i}\eps_{t_{i}}-\frac{1}{\sqrt{N_1}}\sum_{t_{i}\leq 1}\Delta X_{t_i}\eps_{t_{i-1}}\notag.
\end{align}
We have
$$E\left[\left.\left(\frac{1}{\sqrt{N_1}}\sum_{t_{i}\leq 1}\Delta X_{t_i}\eps_{t_{i}}\right)^2\right|\mathcal{F}_1\right]=\frac{\sigma^2_\eps}{N_1}[X,X]_1=O_p\left(\frac{1}{n}\right)$$
by again C(2) and $[X,X]_1=O_p(1)$. The same argument applies to the other term. Hence, $[X,\eps]_1/\sqrt{N_1}=O_p\left(1/\sqrt{n}\right)$. For the last term $([\eps,\eps]_1-2N_1\sigma_\eps^2)/2\sqrt{N_1},$ we rewrite it as
\begin{align}
\frac{([\eps,\eps]_1-2N_1\sigma_\eps^2)}{2\sqrt{N_1}}=\frac{1}{\sqrt{N_1}}\sum_{t_i\leq 1}(\eps_{t_i}^2-\sigma_\eps^2)-\frac{1}{\sqrt{N_1}}\sum_{t_{i}\leq 1}\eps_{t_{i-1}}\eps_{t_{i}}-\frac{\eps_{t_0}^2+\eps_{t_{N_1}}^2-2\sigma_\eps^2}{2\sqrt{N_1}}\notag.
\end{align}
Similarly as above, we have
$$
E\left[\left.\left(\frac{1}{\sqrt{N_1}}\sum_{t_i\leq 1}(\eps_{t_i}^2-\sigma_\eps^2)\right)^2\right|\mathcal{F}_1\right]=\left(1+\frac{1}{N_1}\right)
\var(\eps^2)=O_p(1),
$$
and $1/\sqrt{N_1}\sum_{t_{i}\leq 1}\eps_{t_{i-1}}\eps_{t_{i}}=O_p(1)$ and $(\eps_{t_0}^2+\eps_{t_{N_1}}^2-2\sigma_\eps^2)/(2\sqrt{N_1})=O_p(1/\sqrt{n})$, completing the proof.
\end{proof}

Next, as we will deal with sums of a random number of random variables repeatedly, the following simple lemma turns out to be very useful.
\begin{Lem}\label{lem:rnd_sum}
Suppose that $N$ is a random variable taking values in nonnegative integers, and $X_1,X_2,\ldots$ are nonnegative random variables satisfying
\[
  E(X_iI_{\{i\leq N\}})\leq C\cdot P(i\leq N), \q\mbox{for all }~i.
\]
Then
\[
  E\sum_{i=1}^N X_i \leq C\cdot E(N).
\]
\end{Lem}
\begin{proof}
The conclusion follows from the fact that $\sum_{i=1}^N X_i =\sum_{i=1}^\infty X_i I_{\{i\leq N\}}$ and the Monotone Convergence Theorem.
\end{proof}

\subsection{Proof of Theorem 1: single sub-grid case}\label{ssec:pf_single_grid}
The basic idea is to decompose
$$\widehat{\langle X,X\rangle}_t^{LA}-\langle X,X\rangle_t=[\overline{Y},\overline{Y}]^{\mathcal{S}}_t-\frac{2L_t}{p}\widehat{\sigma_{\eps}^2}-\langle X,X\rangle_t$$ into existing familiar quantities and other negligible terms. The proof is divided into three steps.

\subsection*{Step 1: Introducing $\widetilde{Y}$}
The local average can be decomposed as follows
\begin{align}
\overline{Y}_{t_{i,0}}&=\frac{1}{p}\sum_{j=0}^{p-1}(X_{t_{iq+p-j}}+\eps_{t_{iq+p-j}})\notag\\
&=X_{t_{i,0}}-\frac{1}{p}\sum_{j=0}^{p-1}(X_{t_{iq+p}}-X_{t_{iq+p-j}})+\frac{1}{p}\sum_{j=0}^{p-1}\eps_{t_{iq+p-j}}\notag\\
&=X_{t_{i,0}}
-\sum_{j=2}^{p}\frac{j-1}{p}\Delta X_{t_{iq+j}}
+\bar{\eps}_{t_{i,0}}.\notag
\end{align}
where
$$\bar{\eps}_{t_{i,0}}:=\frac{1}{p}\sum_{j=0}^{p-1}\eps_{t_{iq+p-j}},$$
which is a sequence of independent random variables with common mean $E\bar{\eps}=0$, variance $E\bar{\eps}^2=\sigma_\eps^2/p,$ $E\bar{\eps}^3=E\eps^3/p^2$ and $E\bar{\eps}^4=E\eps^4/p^3+3(p-1)(\sigma_\eps^2)^2/p^3.$ Motivated by the above decomposition, we introduce the new process $\widetilde{Y}$ as follows
\begin{align}
\widetilde{Y}_{t_{i,0}}=X_{t_{i,0}}+\bar{\eps}_{t_{i,0}},~{\rm for}~i=0,\ldots,L_1.\notag
\end{align}
The strategy is that if the difference $([\widetilde{Y},\widetilde{Y}]^{\mathcal{S}}_t-[\overline{Y},\overline{Y}]^{\mathcal{S}}_t)$, where similarly to the definition of $[\overline{Y},\overline{Y}]^{\mathcal{S}}_t$
$$
[\widetilde{Y},\widetilde{Y}]^{\mathcal{S}}_t:=\sum_{t_{i,0}\leq t}(\Delta\widetilde{Y}_{t_{i,0}})^2~{\rm and}~\Delta\widetilde{Y}_{t_{i,0}}=\widetilde{Y}_{t_{i,0}}-\widetilde{Y}_{t_{i-1,0}},
$$
is of a negligible order, then one needs only to deal with $[\widetilde{Y},\widetilde{Y}]_t^{\mathcal{S}}$.

\subsection*{Step 2: Determining the order of $([\widetilde{Y},\widetilde{Y}]^{\mathcal{S}}_t-[\overline{Y},\overline{Y}]^{\mathcal{S}}_t)$}
For notational convenience, we define for $k=0,1,\ldots, q-1$,
\begin{equation}
\left\{
\begin{array}{ll}
 A_i^k:=X_{t^k_{i,0}},\\
 B_i^k:=\bar{\eps}_{t^k_{i,0}}, \q {\rm and}\\ C_i^k:=
-\sum_{j=2}^{p}\frac{j-1}{p}\Delta X_{t_{iq+j+k}},
 \end{array}
    \right.\label{eq:ABCdefinitions}
\end{equation}
and let
\[
A_i=A_i^0, \q B_i=B_i^0,\q\mbox{ and } C_i=C_i^0.
\]
Adopting the above notation, we can write
\begin{align}
[\overline{Y},\overline{Y}]^{\mathcal{S}}_t-[\widetilde{Y},\widetilde{Y}]_t^{\mathcal{S}}&=\sum_{t_{i,0}\leq t}(\Delta A_i+\Delta B_i+\Delta C_i)^2-\sum_{t_{i,0}\leq t}(\Delta A_i+\Delta B_i)^2\notag\\
&=\underbrace{\sum_{t_{i,0}\leq t}\Delta C_i^2}_{I}+\underbrace{2\sum_{t_{i,0}\leq t}\Delta A_i\Delta C_i}_{II}+\underbrace{2\sum_{t_{i,0}\leq t}\Delta B_i\Delta C_i}_{III}.\label{eq:diffY_decomp}
\end{align}
By Cauchy-Schwartz inequality, for any $t$,
\begin{align}
I\leq 4\sum_{t_{i,0}\leq 1}C_i^2\leq 4\sum_{t_{iq}\leq 1} C_i^2\notag.
\end{align}
By the BDG inequality and the strong markov property of $X$,
\begin{align}
E[C_i^2I_{\{t_{iq}\leq 1\}}]&=E\left(I_{\{t_{iq}\leq 1\}}E\left[\left.C_i^2\right|\mathcal{F}_{t_{iq}}\right]\right)
\leq CE\left[I_{\{t_{iq}\leq 1\}}E\left(\left.\sum_{j=1}^{p-1}\frac{j^2}{p^2}\int_{t_{iq+j}}^{t_{iq+j+1}}\sigma_s^2ds\right|\mathcal{F}_{t_{iq}}\right)\right]\notag\\
&\leq CE\left[I_{\{t_{iq}<1\}}E\left(\left.\sum_{j=1}^{p-1}\frac{j^2}{p^2}\sigma_+^2\Delta_n\right|\mathcal{F}_{t_{iq}}\right)\right]\leq C\frac{p}{n^{1-\eta}} P(t_{iq}\leq 1)\notag.
\end{align}
By Lemma \ref{lem:rnd_sum} and the fact that $N_1\leq Cn$ and hence $L_1\leq Cn/q$ we then obtain
\begin{equation}\label{eq:termIiterationofExpectation}
E(I) \leq 4E\left(\sum_{t_{i,0}\leq 1}C_i^2\right)\leq 4E\left(\sum_{t_{iq}\leq 1} C_i^2\right)\leq C p/(qn^{-\eta}).
\end{equation}
Next we study term $III.$ In fact,
\begin{align}
E(III)^2
=4E\left[E\left(\left.\left(\sum_{t_{i,0}\leq t}\Delta B_i\Delta C_i\right)^2\right|\mathcal{F}_1\right)\right]
&\leq C\frac{\sigma_\eps^2}{p}E\left(\sum_{t_{i,0}\leq t} C_i^2\right).\label{eq:termIfortheoremIIuse}
\end{align}
Hence, it follows from (\ref{eq:termIiterationofExpectation}) that $III=O_p\left((1/(qn^{-\eta}))^{1/2}\right)$.

Finally we deal with term $II$.
\begin{Cla}
$II=2\sum_{t_{i,0}\leq t}\Delta A_i\Delta C_i= O_p\left(\frac{\ell p}{n^{1-\eta}}\right)+O_p\left(\sqrt{\frac{p}{n^{1-2\eta}}}\right).$
\end{Cla}
{\it Proof of the Claim.} First notice that
$$
\sum_{t_{i,0}\leq t}\Delta A_i\Delta C_i=\sum_{t_{i,0}\leq t}C_{i}\Delta A_i-\sum_{t_{i,0}\leq t}C_{i-1}\Delta A_i,
$$
where, by BDG inequality and (\ref{eq:termIiterationofExpectation}), we have that
\begin{equation}\label{eq:termIIfortheoremIIuse}
E\left(\sum_{t_{i,0}\leq t}C_{i-1}\Delta A_i\right)^2
\leq C E\sum_{t_{i,0}\leq t}C_{i-1}^2 \frac{q}{n^{1-\eta}}\sigma_+^2
\leq C \ell\frac{pq}{n^{2-2\eta}}\leq C\frac{p}{n^{1-2\eta}},
\end{equation}
and hence
$
\sum_{t_{i,0}\leq t}C_{i-1}\Delta A_i=O_p\left(\sqrt{\frac{p}{n^{1-2\eta}}}\right).
$
Next define $\Delta X^{(i)}:=X_{t_{iq+1}}-X_{t_{(i-1)q+p}}$. Then
\begin{align}
\sum_{t_{i,0}\leq t}C_{i}\Delta A_{i}
&=-\underbrace{\sum_{t_{i,0}\leq t}\sum_{j=2}^{p}\frac{j-1}{p}\left(\Delta X_{t_{iq+j}}\right)^2}_{\varsigma_1}\label{eq:II_decomp}\\
&-\underbrace{\sum_{t_{i,0}\leq t}\sum_{j=2}^{p}\left[\frac{j-1}{p}\Delta X^{(i)}+\frac{1}{p}\sum_{m=2}^{j-1}(j+m-2)\Delta X_{t_{iq+m}}\right]\Delta X_{t_{iq+j}}}_{\varsigma_2}.\notag
\end{align}
By BDG inequality, $E\varsigma_1\leq C\ell p/n^{1-\eta}$; moreover, by BDG inequality again,
$$
E\left[\frac{j-1}{p}\Delta X^{(i)}+\frac{1}{p}\sum_{m=2}^{j-1}(j+m-2)\Delta X_{t_{iq+m}}\right]^2\leq Cq/n^{1-\eta},
$$
hence, applying once more the BDG inequality one obtains that
\begin{align}
E\left(\varsigma_2\right)^2\leq C\frac{\ell pq}{n^{2-2\eta}}.\label{eq:termIIIfortheoremIIIuse}
\end{align}
It follows that $\varsigma_2=O_p\left((p/n^{1-2\eta})^{1/2}\right)$ and moreover,
$$
II=O_p\left(\frac{\ell p}{n^{1-\eta}}\right)+O_p\left(\sqrt{\frac{p}{n^{1-2\eta}}}\right).
$$
\qed

To summarize,
\begin{align}
[\overline{Y},\overline{Y}]^{\mathcal{S}}_t-[\widetilde{Y},\widetilde{Y}]_t^{\mathcal{S}}=O_p\left(\frac{\ell p}{n^{1-\eta}}\right)+O_p\left(\sqrt{\frac{p}{n^{1-2\eta}}}\right).\label{eq:barYtildeYdiff}
\end{align}

\begin{Rem}
In the proof for Theorem 2 below, we will analyze $[\overline{Y},\overline{Y}]^{\mathcal{S}}_t-[\widetilde{Y},\widetilde{Y}]_t^{\mathcal{S}}$
in more detail. Notice that $I=2\sum_{t_{i,0}\leq t}C_i^2-2\sum_{t_{i,0}\leq t}C_{i-1}C_{i}-C_0^2-C_{L_t}^2$, where the end effect terms
$C_0^2$ and $C_{L_t}^2$ are $O_p\left((p/n^{1-\eta})^{1/2}\right)$ and by BDG inequality, $\sum_{t_{i,0}\leq t}C_{i-1}C_{i}$ is
$O_p\left(p\ell^{1/2}/n^{1-\eta}\right)$. Hence, from (\ref{eq:diffY_decomp}) and the analysis of terms $I$, $II$ and $III$,
\begin{align}
[\overline{Y},\overline{Y}]^{\mathcal{S}}_t-[\widetilde{Y},\widetilde{Y}]_t^{\mathcal{S}}&=2\sum_{t_{i,0}\leq t}\left(\sum_{j=2}^{p}\frac{j-1}{p}\Delta X_{t_{iq+j}}\right)^2-2\sum_{t_{i,0}\leq t}\sum_{j=2}^{p}\frac{j-1}{p}(\Delta X_{t_{iq+j}})^2+O_p\left(\sqrt{\frac{p}{n^{1-2\eta}}}\right).\label{eq:single_to_mov1}
\end{align}
Moreover,
\begin{align}
\left(\sum_{j=2}^{p}\frac{j-1}{p}\Delta X_{t_{iq+j}}\right)^2=\sum_{j=2}^{p}\frac{(j-1)^2}{p^2}(\Delta X_{t_{iq+j}})^2+2\sum_{2\leq k<j\leq p}\frac{(k-1)(j-1)}{p^2}\Delta X_{t_{iq+k}}\Delta X_{t_{iq+j}}.\label{eq:single_to_mov2}
\end{align}
\label{Rem:moredetail}
\end{Rem}
\subsection*{Step 3: CLT for $\widehat{\langle X,X\rangle}_t^{LA}$}
We first notice that
\begin{align}
[\widetilde{Y},\widetilde{Y}]^{\mathcal{S}}_t
&=\sum_{t_{i,0}\leq t}(\Delta X_{t_{i,0}})^2+2\sum_{t_{i,0}\leq t}(\Delta X_{t_{i,0}})(\Delta \bar{\eps}_{t_{i,0}})+\sum_{t_{i,0}\leq t}(\Delta \bar{\eps}_{t_{i,0}})^2\\
&:=[X,X]_t^{\mathcal{S}}+2[X,\bar{\eps}]^{\mathcal{S}}_t+[\bar{\eps},\bar{\eps}]_t^{\mathcal{S}}\notag\\
\notag
\end{align}
Hence, we have the following decomposition
\begin{align}
\widehat{\langle X,X\rangle}_t^{LA}-\langle X,X\rangle_t&=[\overline{Y},\overline{Y}]^{\mathcal{S}}_t-[\widetilde{Y},\widetilde{Y}]^{\mathcal{S}}_t+[\widetilde{Y},\widetilde{Y}]^{\mathcal{S}}_t-\langle X,X\rangle_t-\frac{2L_t}{p}\widehat{\sigma_{\eps}^2}\notag\\
&=\left([\overline{Y},\overline{Y}]^{\mathcal{S}}_t-[\widetilde{Y},\widetilde{Y}]^{\mathcal{S}}_t\right)+\left([X,X]^{\mathcal{S}}_t-\langle X,X\rangle_t\right)\notag\\
&+\left([\bar{\eps},\bar{\eps}]_t^{\mathcal{S}}-\frac{2L_t}{p}\sigma_\eps^2\right)\notag\\
&-\frac{2L_t}{p}\left(\widehat{\sigma_{\eps}^2}-\sigma_\eps^2\right)+2[X,\bar{\eps}]^{\mathcal{S}}_t.\label{eq:finalDecomp}
\end{align}
Recall that $\ell\sim C_\ell n^{\alpha}$ and $p\sim C_pn^{\alpha}$. Then by (\ref{eq:barYtildeYdiff}), $[\overline{Y},\overline{Y}]^{\mathcal{S}}_t-[\widetilde{Y},\widetilde{Y}]^{\mathcal{S}}_t$ is $o_p(1/\sqrt{\ell})$.
As to the term $2L_t/p(\widehat{\sigma_{\eps}^2}-\sigma_\eps^2)$ in (\ref{eq:finalDecomp}),  by Lemma \ref{Lem:ConvgcOfsigeps} together with C(2) and C(4), we have that
\[
\frac{2L_t}{p}\left(\widehat{\sigma_{\eps}^2}-\sigma_\eps^2\right)
=\frac{2L_t}{p\sqrt{N_1}}\sqrt{N_1}\left(\widehat{\sigma_{\eps}^2}-\sigma_\eps^2\right)
=O_p\left(\frac{\ell}{p\sqrt{n}}\right)=o_p\left(\frac{1}{\sqrt{\ell}}\right)~{\rm in~}D[0,1].
\]

Therefore, in order to prove the asymptotic property of $\sqrt{\ell}\left(\widehat{\langle X,X\rangle}_t^{LA}-\langle X,X\rangle_t\right),$ one only needs to prove the FCLT for the following quantity
\begin{align}
\sqrt{\ell}\left([X,X]^{\mathcal{S}}_t-\langle X,X\rangle_t\right)+\sqrt{\ell}\left([\bar{\eps},\bar{\eps}]_t^{\mathcal{S}}-\frac{2L_t}{p}\sigma_\eps^2\right)
+2\sqrt{\ell}[X,\bar{\eps}]^{\mathcal{S}}_t.\label{eq:FinalQant}
\end{align}
Firstly, notice that
$$
[\bar{\eps},\bar{\eps}]_t^{\mathcal{S}}-\frac{2L_t}{p}\sigma_\eps^2
=2\sum_{i=1}^{L_t}\left(\bar{\eps}_{t_{i,0}}^2-\frac{\sigma_\eps^2}{p}\right)
 - \bar{\eps}_{t_{0,0}}^2 - \bar{\eps}_{t_{L_t,0}}^2 -
2\sum_{i=1}^{L_t}\bar{\eps}_{t_{i-1,0}}\bar{\eps}_{t_{i,0}}.
$$
Note that
$\bar{\eps}_{t_{0,0}}^2=O_p(1/p)$, hence $\sqrt{\ell}\bar{\eps}_{t_{0,0}}^2=o_p(1)$, and so is $\sqrt{\ell}\bar{\eps}_{t_{L_t,0}}^2$.
Moreover,
$$
[X,\bar{\eps}]^{\mathcal{S}}_t
=\sum_{i=1}^{L_t}(\Delta X_{t_{i,0}}-\Delta X_{t_{i+1,0}})\bar{\eps}_{t_{i,0}}
+\Delta X_{t_{L_t+1,0}}\bar{\eps}_{t_{L_t,0}}
-\Delta X_{t_{1,0}}\bar{\eps}_{t_{0,0}}.
$$
Note that $\Delta X_{t_{L_t+1,0}}\bar{\eps}_{t_{L_t,0}}=O_p\left((1/(p\ell n^{-\eta}))^{1/2}\right)$ and so is $\Delta X_{t_{1,0}}\bar{\eps}_{t_{0,0}}$.
We are hence led to study the following martingales
\begin{align}
&M_t:=\sqrt{\ell}\left([X,X]^{\mathcal{S}}_t-\langle X,X\rangle_t\right),\notag\\
&M^{(1)}_t:=\sqrt{\ell}\sum_{i=1}^{L_t}\left(\bar{\eps}_{t_{i,0}}^2-\frac{\sigma_\eps^2}{p}\right),\notag\\
&M^{(2)}_t:=\sqrt{\ell}\sum_{i=1}^{L_t}\bar{\eps}_{t_{i-1,0}}\bar{\eps}_{t_{i,0}},\notag\\
&M^{(3)}_t:=\sqrt{\ell}\sum_{i=1}^{L_t}(\Delta X_{t_{i,0}}-\Delta X_{t_{i+1,0}})\bar{\eps}_{t_{i,0}}.\notag
\end{align}
Then (\ref{eq:FinalQant}) can be rewritten as
\begin{align}
M_t+2M^{(1)}_t-2M^{(2)}_t+2M^{(3)}_t+o_p(1).\label{eq:FFquant}
\end{align}

Simple calculation gives the corresponding predictable variation processes as follows
\[
\aligned
\langle \left.M^{(1)},M^{(1)}\rangle_t\right|\mathcal{F}_1
&=\ell L_t{\rm Var}(\bar{\eps}^2)=\ell L_t\left(E\bar{\eps}^4-(E\bar{\eps}^2)^2\right)\\
&=\ell L_t\left(\frac{E\eps^4}{p^3}+\frac{2(\sigma_\eps^2)^2}{p^2}-\frac{3(\sigma_\eps^2)^2}{p^3}\right)\\
&\overset{p}\rightarrow 2\left(\frac{C_\ell}{C_p}\sigma_\eps^2\right)^2\int_0^tr_sds,\\
\left.\langle M^{(2)},M^{(2)}\rangle_t\right|\mathcal{F}_1
&=\ell\frac{\sigma_\eps^2}{p}\sum_{i=1}^{L_t}\bar{\eps}_{t_{i-1,0}}^2
\overset{p}\rightarrow \left(\frac{C_\ell}{C_p}\sigma_\eps^2\right)^2\int_0^tr_sds,\q\mbox{and}\\
\langle \left.M^{(3)},M^{(3)}\rangle_t\right|\mathcal{F}_1
&=2\ell\frac{\sigma_\eps^2}{p}[X,X]^{\mathcal{S}}_t-2\ell\frac{\sigma_\eps^2}{p}\sum_{i=1}^{L_t}(\Delta X_{t_{i,0}})(\Delta X_{t_{i+1,0}})\notag\\
&~~~-\ell\frac{\sigma_\eps^2}{p}\left((\Delta X_{t_{1,0}})^2+(\Delta X_{t_{L_t+1,0}})^2\right)\notag\\
&\overset{p}\rightarrow 2\frac{C_\ell}{C_p}\langle X,X\rangle_t\sigma_\eps^2,\notag
\endaligned
\]
where in the last convergence we used the fact that $\ell/p\cdot \sum_{i=1}^{L_t}(\Delta X_{t_{i,0}})(\Delta X_{t_{i+1,0}})\to 0$ in $D[0,1]$ since it is a martingale with predictable variation
\[
\frac{\ell^2}{p^2}\sum_{i=1}^{L_t}(\Delta X_{t_{i,0}})^2\int_{t_{i,0}}^{t_{i+1,0}}\sigma_s^2 \ ds
\leq C \frac{\ell^2 q}{p^2 n^{1-\eta}} \sum_{i=1}^{L_t}(\Delta X_{t_{i,0}})^2=O_p\left(\frac{\ell}{p^2n^{-2\eta}}\right)=o_p(1).
\]
Furthermore, the predictable covariation processes of $M^{(1)}, M^{(2)}$ and $M^{(3)}$ are
\[
\aligned
\left.\langle M^{(1)},M^{(2)}\rangle_t\right|\mathcal{F}_1
&=\ell E\bar{\eps}^3\sum_{i=1}^{L_t}\bar{\eps}_{t_{i-1,0}}
=\ell\frac{E\eps^3}{p^2}\sum_{i=1}^{L_t}\bar{\eps}_{t_{i-1,0}}
=O_p\left(\frac{\sqrt{\ell}}{p^{3/2}}\right)=o_p(1),\\
\left.\langle M^{(1)},M^{(3)}\rangle_t\right|\mathcal{F}_1
&=\ell\sum_{i=1}^{L_t}(\Delta X_{t_{i,0}}-\Delta X_{t_{i+1,0}})E(\bar{\eps}^3)\\
&=\ell E(\bar{\eps}^3)(\Delta X_{t_{1,0}}- \Delta X_{t_{L_t+1,0}})
=O_p(\sqrt{\ell/(p^4n^{-\eta})})=o_p(1),\q\mbox{and}\\
\left.\langle M^{(2)},M^{(3)}\rangle_t\right|\mathcal{F}_1
&=\ell E\bar{\eps}^2\sum_{i=1}^{L_t}(\Delta X_{t_{i,0}}-\Delta X_{t_{i+1,0}})\bar{\eps}_{t_{i-1,0}}=O_p\left(\frac{\ell}{p^{3/2}}\right)=o_p(1),\notag
\endaligned
\]
where the last order follows from the fact that $\sum_{i=1}^{L_t}(\Delta X_{t_{i,0}}-\Delta X_{t_{i+1,0}})\bar{\eps}_{t_{i-1,0}}=O_p(1/p^{1/2})$ by considering its predictable variation process similarly to the way that we treat $M^{(3)}_t$. The Lindeberg type condition can be easily verified by using the same calculations as above and the assumption that $(\eps_{t_i})_{i\geq1}$ is an independent sequence with finite forth moment. Therefore, the usual martingale central limit theorem gives
\begin{equation}\label{conv:M_i}
\left.\left(
\begin{array}{cll}
M^{(1)}_t\\
M^{(2)}_t\\
M^{(3)}_t
\end{array}
\right)\right| \mathcal{F}_1\Longrightarrow \Sigma^{1/2} \left(
\begin{array}{cll}
W_1(t)\\
W_2(t)\\
W_3(t)
\end{array}
\right),
\end{equation}
where $W_1,$ $W_2$ and $W_3$ are independent standard Brownian motions and the limiting covariance matrix process is given by
\begin{align}
\Sigma=
\left(
\begin{array}{cll}
2\left(\frac{C_\ell}{C_p}\sigma_\eps^2\right)^2\int_0^tr_sds&0&0\\
0&\left(\frac{C_\ell}{C_p}\sigma_\eps^2\right)^2\int_0^tr_sds&0\\
0&0&2\frac{C_\ell}{C_p}\langle X,X\rangle_t\sigma_\eps^2
\end{array}
\right).\notag
\end{align}

Finally, by Theorem 1 in \citet{lmrzz09}, we have the following convergence for $M_t$
\begin{equation}\label{conv:M}
M_t \Longrightarrow \frac{2}{3}\int_0^tv_s\sigma_s dX_s + \int_0^t\left(\frac{2}{3}u_s-\frac{4v_s^2}{9}\right)\sigma_s^4\ dW(s),
\end{equation}
where $W(s)$ is a standard Brownian motion. Furthermore, it is easy to see that $\langle M,M^{(i)}\rangle_t=0$ for $i=1,2,3$, hence $W(t)$ is independent of $W_i(t), i=1,2,3$. Combining this fact with \eqref{conv:M_i} and \eqref{conv:M} yields the desired convergence.
\qed


\subsection{Proof of Theorem 2: multiple sub-grids case}\label{ssec:pf_multiple_grid}
We shall establish the following stable in law convergence
\begin{align}
&~~~\sqrt{\ell}\left(\frac{1}{q}\sum_{k=0}^{q-1}[\overline{Y},\overline{Y}]^{\mathcal{S}_k}_t-\frac{2N_t}{pq}\widehat{\sigma_{\eps}^2}-(1+A(p,q))\langle X,X\rangle_t\right)\notag\\
&\Longrightarrow \frac{2}{3}\int_0^t\bar{v}_s\sigma_sdX_s+\int_0^t\left[\left(4w_s-\frac{4}{9}\bar{v}_s^2\right)\sigma_s^4+\frac{8C_\ell^3}{C_p}r_s(\sigma_\eps^2)^2\right]^{1/2}dB_s\notag.
\end{align}
Similar to the convention of using notation $[\overline{Y},\overline{Y}]^{\mathcal{S}_k}_t$ to denote RV of local averaged $Y$ process computed based on the $k$th sub-grid $\mathcal{S}_k,$ all subsequent notations in the proof with superscript $k$ or $\mathcal{S}_k$ indicate that the same operation as performed on the sub-grid $\mathcal{S}\equiv\mathcal{S}_0$ is applied to the $k$th sub-grid.

The proof for Theorem 2 also proceeds in three steps. Similar to the proof of Theorem 1, the proof for Theorem 2 is based on the following decomposition
\begin{align}
&~~~\sqrt{\ell}\left(\frac{1}{q}\sum_{k=0}^{q-1}[\overline{Y},\overline{Y}]^{\mathcal{S}_k}_t-\frac{2N_t}{pq}\widehat{\sigma_{\eps}^2}-(1+A(p,q))\langle X,X\rangle_t\right)\notag\\
&=\underbrace{\sqrt{\ell}\left(\frac{1}{q}\sum_{k=0}^{q-1}[\overline{Y},\overline{Y}]^{\mathcal{S}_k}_t-\frac{1}{q}\sum_{k=0}^{q-1}[\widetilde{Y},\widetilde{Y}]^{\mathcal{S}_k}_t-A(p,q)\langle X,X\rangle_t\right)}_{ I}\notag\\
&~~~+\underbrace{\sqrt{\ell}\left(\frac{2}{q}\sum_{k=0}^{q-1}[X,\bar\eps]_t^{\mathcal{S}_k}
+\frac{1}{q}\sum_{k=0}^{q-1}[\bar\eps,\bar\eps]_t^{\mathcal{S}_k}-\frac{2N_t}{pq}\widehat{\sigma_{\eps}^2}\right)}_{ II}+\underbrace{\sqrt{\ell}\left(\frac{1}{q}\sum_{k=0}^{q-1}[X,X]^{\mathcal{S}_k}_t-\langle X,X\rangle_t\right)}_{
III}\notag.
\end{align}
Assuming $\ell\sim C_\ell n^{\alpha}$ and $p\sim C_pn^{3\alpha-1}$ with assumptions made in the theorem on $\alpha$ and $\eta$, we shall show in Step 1 that $ I=o_p(1)$; in Step 2 that $ II$ satisfies a martingale CLT; in Step 3 a CLT with asymptotic bias decomposition for term $III$; and, finally, sum up in Step 4 .

\subsection*{Step 1}
To show $I=o_p(1)$, we consider the difference
\begin{align}
&~~~~~~\frac{\sqrt{\ell}}{q}\sum_{k=0}^{q-1}[\overline{Y},\overline{Y}]^{\mathcal{S}_k}_t-\frac{\sqrt{\ell}}{q}\sum_{k=0}^{q-1}[\widetilde{Y},\widetilde{Y}]_t^{\mathcal{S}_k}\notag\\
&=\frac{\sqrt{\ell}}{q}\sum_{k=0}^{q-1}\left(\sum_{t^k_{i,0}\leq t}(\Delta C^k_i)^2+2\sum_{t^k_{i,0}\leq t}\Delta A^k_i\Delta C^k_i+2\sum_{t^k_{i,0}\leq t}\Delta B^k_i\Delta C^k_i\right)\label{eq:mov_diff},
\end{align}
adopting the previous notational convention for the single sub-grid case, where $A_i^k,$ $B_i^k$ and $C_i^k$ are defined in (\ref{eq:ABCdefinitions}). Roughly speaking, recall (\ref{eq:single_to_mov1}) and (\ref{eq:single_to_mov2}) of Remark \ref{Rem:moredetail} from the end of Step 2 in the proof for Theorem 1, we expect  the difference (\ref{eq:mov_diff}) to be
\begin{align}
&~~~\frac{2\sqrt{\ell}}{q}\sum_{j=1}^{p-1}\left(\frac{j^2}{p^2}-\frac{j}{p}\right)\cdot\sum_{t_{i}\leq t}(\Delta X_{t_i})^2+o_p(1)\notag\\
&=\sqrt{\ell}A(p,q)[X,X]_t+o_p(1).\label{eq:MAdiff}
\end{align}
It is easy to see that $\sqrt{\ell}A(p,q)([X,X]_t-\langle X,X\rangle_t)=o_p(1)$. Hence $I=o_p(1)$ if we can show that \eqref{eq:MAdiff} holds.

We now verify \eqref{eq:MAdiff}. It is easy to see that the RHS of \eqref{eq:mov_diff} equals
\[
\aligned
&\underbrace{\frac{2\sqrt{\ell}}{q}\sum_{k=0}^{q-1}\sum_{t^k_{i,0}\leq t}(C_i^k)^2}_{I.i} + ~\underbrace{\frac{2\sqrt{\ell}}{q}\sum_{k=0}^{q-1}\sum_{t^k_{i,0}\leq t}C^k_{i-1}C^k_{i}}_{I.ii}
+~\underbrace{\frac{2\sqrt{\ell}}{q}\sum_{k=0}^{q-1}\sum_{t^k_{i,0}\leq t}\Delta B^k_i\Delta C_i^k}_{I.iii};\\
&-\underbrace{\frac{2\sqrt{\ell}}{q}\sum_{k=0}^{q-1}\sum_{t^k_{i,0}\leq t}C_{i-1}^k\Delta A^k_i}_{I.iv}
+~\underbrace{\frac{2\sqrt{\ell}}{q}\sum_{k=0}^{q-1}\sum_{t^k_{i,0}\leq t}C^k_{i}\Delta A^k_i}_{I.v}
+ o_p(1).
\endaligned
\]
We analyze them one by one.

We start with $I.i=\frac{2\sqrt{\ell}}{q}\sum_{k=0}^{q-1}\sum_{t^k_{i,0}\leq t}(C_i^k)^2.$
Notice that on each sub-grid $\mathcal{S}_k,$
$$(C_i^k)^2=\sum_{j=1}^{p-1}\frac{j^2}{p^2}(\Delta X_{t_{iq+j+k+1}})^2
+2\sum_{j=2}^{p-1}\left(\sum_{m=1}^{j-1}\frac{m}{p}\Delta X_{t_{iq+m+k+1}}\right)\frac{j}{p}\Delta X_{t_{iq+j+k+1}}.$$
Therefore, term $I.i$ can be rewritten as follows
\begin{align}
&\underbrace{\frac{2\sqrt{\ell}}{q}\left(\sum_{j=1}^{p-1}\frac{j^2}{p^2}\right)\sum_{t_{i}\leq t}(\Delta X_{t_i})^2}_{{\rm dominating ~term~} A}
-\frac{2\sqrt{\ell}}{q}\left(\sum_{i=1}^{p-1}\left(\sum_{j=i}^{p-1}\frac{j^2}{p^2}\right)(\Delta X_{t_i})^2+
\sum_{i=L_tq+2}^{L_tq+p-1}\left(\sum_{j=1}^{i-L_tq-1}\frac{j^2}{p^2}\right)(\Delta X_{t_{i+1}})^2\right)
\notag\\
&+\left(\frac{4\sqrt{\ell}}{q}\sum_{i=2}^{p-2}\left(\sum_{j=2}^i\frac{j}{p}\sum_{m=1}^{j-1}\frac{m}{p}\Delta X_{t_{i-j+m+1}}\right)\Delta X_{t_{i+1}}+\frac{4\sqrt{\ell}}{q}\sum_{i=p-1}^{L_tq+2}\left(\sum_{j=2}^{p-1}\frac{j}{p}\sum_{m=1}^{j-1}\frac{m}{p}\Delta X_{t_{i-j+m+1}}\right)\Delta X_{t_{i+1}}\right.\notag\\
&~~~\q\q\left.+\frac{4\sqrt{\ell}}{q}\sum_{i=L_tq+3}^{L_tq+p-1}\left(\sum_{j=i-L_tq}^{p-1}\frac{j}{p}\sum_{m=1}^{j-1}\frac{m}{p}\Delta X_{t_{i-j+m+1}}\right)\Delta X_{t_{i+1}}\right)\\
&:={\rm dominating ~term~} A - \mbox{edge term } B + S^{(1)}_t\notag.
\end{align}
It is easy to see that the edge term $B=o_p(1)$.
We shall further show that $S^{(1)}_t$ is negligible. To see that, notice that its
 expected predictable variation satisfies
\begin{align}
E\langle S^{(1)},S^{(1)}\rangle_1
&\leq \frac{C\ell\sigma_+^2}{q^2 n^{1-\eta}}E\sum_i\left(\sum_{j=2}^{p-1}\frac{j}{p}\sum_{m=1}^{j-1}\frac{m}{p}\Delta X_{t_{i-j+m+1}}\right)^2\notag\\
&=\frac{C\ell\sigma_+^2}{q^2 n^{1-\eta}}E\sum_i\left(\sum_{m=1}^{p-2} \left(\sum_{j=m+1}^{p-1}\frac{j}{p}\frac{j-m}{p}\right)\Delta X_{t_{i-m+1}}\right)^2\notag\\
&\leq \frac{C\ell p^3}{n^{1-2\eta}q^2} ,\notag
\end{align}
which follows from the fact that
$$
E\left(\sum_{m=1}^{p-2}\left(\sum_{j=m+1}^{p-1}\frac{j}{p}\frac{j-m}{p}\right)\Delta X_{t_{i-m+1}}\right)^2\leq C\frac{p^3}{n^{1-\eta}},~~~{\rm uniformly~in}~i.
$$
Therefore,
$$S^{(1)}_t=O_p\left(\sqrt{\frac{\ell p^3}{n^{1-2\eta}q^2}}\right) =O_p\left(\sqrt{\frac{p^3}{q^3n^{-2\eta}}}\right)=o_p(1)~{\rm in}~D[0,1].$$

Next we estimate $I.ii=\frac{2\sqrt{\ell}}{q}\sum_{k=0}^{q-1}\sum_{t^k_{i,0}\leq t}C^k_{i-1}C^k_{i}.$ It can be rearranged as
\begin{align}
I.ii&=\frac{2\sqrt{\ell}}{q}\sum_{i=q+1}^{q+p-2}\left(\sum_{j=1}^{i-q}\frac{j}{p}\sum_{m=1}^{p-1}\frac{m}{p}\Delta X_{t_{i-q-j+m+1}}\right)\Delta X_{t_{i+1}}\notag\\
&~~~+\frac{2\sqrt{\ell}}{q}\sum_{i=q+p-1}^{L_tq+1}\left(\sum_{j=1}^{p-1}\frac{j}{p}\sum_{m=1}^{p-1}\frac{m}{p}\Delta X_{t_{i-q-j+m+1}}\right)\Delta X_{t_{i+1}}\notag\\
&~~~+\frac{2\sqrt{\ell}}{q}\sum_{i=L_tq+2}^{L_tq+p-1}\left(\sum_{j=i-L_tq}^{p-1}\frac{j}{p}\sum_{m=1}^{p-1}\frac{m}{p}\Delta X_{t_{i-q-j+m+1}}\right)\Delta X_{t_{i+1}}.\notag
\end{align}
We denote the above quantity as $S^{(2)}_t.$ Similar to the treatment for $I.i$,
\begin{align}
E\langle S^{(2)},S^{(2)}\rangle_1
&\leq \frac{C\ell\sigma_+^2}{q^2 n^{1-\eta}}E\sum_i\left(\sum_{j=1}^{p-1}\frac{j}{p}\sum_{m=1}^{p-1}\frac{m}{p}\Delta X_{t_{i-q-j+m+1}}\right)^2
\leq\frac{C}{q^2}\frac{\ell p^3}{n^{1-2\eta}}.\notag
\end{align}
Therefore,
$$S^{(2)}_t=
O_p\left(\sqrt{\frac{\ell p^3}{n^{1-2\eta}q^2}}\right) =O_p\left(\sqrt{\frac{p^3}{q^3n^{-2\eta}}}\right)=o_p(1)~{\rm in}~D[0,1].$$

Now we study $I.iii=\frac{2\sqrt{\ell}}{q}\sum_{k=0}^{q-1}\sum_{t^k_{i,0}\leq t}\Delta B^k_i\Delta C_i^k.$
Noticing that the estimate in (\ref{eq:termIfortheoremIIuse}) holds uniformly for sub-grids $\mathcal{S}_k$, hence by Cauchy-Schwartz inequality we obtain that
$$I.iii=O_p\left(\left(\ell/(qn^{-\eta})\right)^{1/2}\right)=o_p(1)~{\rm in~} D[0,1].$$

Now we come to $I.iv=\frac{2\sqrt{\ell}}{q}\sum_{k=0}^{q-1}\sum_{t^k_{i,0}\leq t}C_{i-1}^k\Delta A^k_i.$
By Cauchy-Schwartz inequality again, as the estimate in (\ref{eq:termIIfortheoremIIuse}) holds uniformly for sub-grids $\mathcal{S}_k$, we have
$$I.iv=O_p\left(\left(\ell p/n^{1-2\eta}\right)^{1/2}\right)=o_p(1)~{\rm in~}D[0,1].$$

Finally we deal with $I.v=\frac{2\sqrt{\ell}}{q}\sum_{k=0}^{q-1}\sum_{t^k_{i,0}\leq t}C^k_{i}\Delta A^k_i.$
Similar to the decomposition \eqref{eq:II_decomp} we have
$$
I.v= -\frac{2\sqrt{\ell}}{q}\sum_{k=0}^{q-1}(\varsigma_1^k+\varsigma_2^k)
$$
where, with $\Delta X^{(k,i)}:=X_{t_{iq+k+1}}-X_{t_{{i-1}q+k+p}}$,
\[
\aligned
\varsigma_1^k&=
\sum_{t_{i,0}^k\leq t}\sum_{j=2}^{p}\frac{j-1}{p}\left(\Delta X_{t_{iq+k+j}}\right)^2,\\
\varsigma_2^k&=\sum_{t_{i,0}^k\leq t}\sum_{j=2}^{p}\left[\frac{j-1}{p}\Delta X^{(k,i)}+\frac{1}{p}\sum_{m=2}^{j-1}(j+m-2)\Delta X_{t_{iq+k+m}}\right]\Delta X_{t_{iq+k+j}}.
\endaligned
\]
It is easy to see that
$$
-2\sqrt{\ell}/q\sum_{k=0}^{q-1}\varsigma_1^k
=\underbrace{-2\sqrt{\ell}/q\sum_{j=1}^{p-1}\frac{j}{p}\sum_{t_i\leq t}(\Delta X_{t_i})^2}_{{\rm dominating~term}~B} + o_p(1).
$$
We next prove that $2\sqrt{\ell}/q\sum_{k=0}^{q-1}\varsigma_2^k$ is negligible.
In fact, the estimate in (\ref{eq:termIIIfortheoremIIIuse}) holds uniformly for all the sub-grids, hence by Cauchy-Schwartz inequality again we get that
$$2\sqrt{\ell}/q\sum_{k=0}^{q-1}\varsigma_2^k=O_p\left(\left(\ell p/n^{1-2\eta}\right)^{1/2}\right)=o_p(1)~{\rm in~}D[0,1].$$

Summing up the computations for $I.i$ to $I.v$, we see that the two dominating terms appearing in $I.i$ and $I.v$ together give the first term in (\ref{eq:MAdiff}) and the rest gives the $o_p(1)$ term in (\ref{eq:MAdiff}).

\subsection*{Step 2}
Now we deal with the term $II$, starting with $2\sqrt{\ell}/q\sum_{k=0}^{q-1}[X,\bar\eps]_t^{\mathcal{S}_k}$.
Denote
$$\Delta_q X_{t_i}=
X_{t_{i}}-X_{t_{i-q}}.
$$
Combining terms with common factor $\eps_{t_i}$ and ordering them chronologically (according to the sequence $(\eps_{t_i})_{i\geq 1}$) we get
\begin{align}
\frac{2\sqrt{\ell}}{q}\sum_{k=0}^{q-1}[X,\bar\eps]_t^{\mathcal{S}_k}&=\frac{2\sqrt{\ell}}{q}\sum_{k=0}^{q-1}\left[\sum_{i=1}^{L_t^k-1}\left(\Delta X_{t^k_{i,0}}-\Delta X_{t^k_{i+1,0}}\right)\bar{\eps}_{t^k_{i,0}}+\Delta X_{t^k_{L_t^k,0}}\bar{\eps}_{t^k_{L_t^k,0}}-\Delta X_{t^k_{1,0}}\bar{\eps}_{t^k_{0,0}}\right]\notag\\
&=\frac{2\sqrt{\ell}}{qp}\sum_{i=q+p}^{(L_t-1)q+1}
\left(\sum_{j=0}^{p-1}
\left[\Delta_q X_{t_{i+j}}-\Delta_qX_{t_{i+q+j}}\right]\right)\eps_{t_i}
+
{\rm remainder},\notag
\end{align}
where the remainder term is a sum similar as above over the $i$'s smaller than $q+p$, and can be easily shown to be $o_p(1)$.
We shall further show that the first summand is also negligible, as follows
\begin{align}
&~~~{\rm Var}\left(\left.\frac{2\sqrt{\ell}}{qp}\sum_{i=q+p}^{(L_t-1)q+1}
\left(\sum_{j=0}^{p-1}\left[\Delta_q X_{t_{i+j}}-\Delta_qX_{t_{i+q+j}}\right]\right)\eps_{t_i}\right|\mathcal{F}_1\right)\notag\\
&=\frac{4\ell\sigma^2_\eps}{q^2p^2}\sum_i\left(\sum_{j=0}^{p-1}\left[\Delta_q X_{t_{i+j}}-\Delta_qX_{t_{i+q+j}}\right]\right)^2\notag\\
&=\frac{4\ell\sigma^2_\eps}{q^2p^2}\sum_i\sum_{j=0}^{p-1}\left((\Delta_q X_{t_{i+j}})^2+(\Delta_qX_{t_{i+q+j}})^2\right)\notag
-\frac{8\ell\sigma^2_\eps}{q^2p^2}\sum_i\sum_{j=0}^{p-1}\Delta_q X_{t_{i+j}}\Delta_qX_{t_{i+q+j}}\notag\\
&~~~+\frac{8\ell\sigma^2_\eps}{q^2p^2}\sum_i\sum_{0\leq j<k\leq p-1}\left[\Delta_q X_{t_{i+j}}-\Delta_qX_{t_{i+q+j}}\right]
\left[\Delta_q X_{t_{i+j}}-\Delta_qX_{t_{i+q+k}}\right]\\
:&= V_1 + V_2 + V_3.
\end{align}
We have, firstly, by applying Lemma \ref{lem:rnd_sum} and using the fact that $E(\Delta_qX_{t_i})^2\leq Cq/n^{1-\eta}$ for all $i,$
\[
EV_1
\leq\frac{C\ell}{q^2p^2} \cdot n\cdot p \frac{q}{n^{1-\eta}}\notag
=\frac{C\ell}{qpn^{-\eta}}\rightarrow0.\notag
\]
This, together with the Cauchy-Schwartz inequality, imply that $E|V_2|\leq C\ell/(pqn^{-\eta})\to 0.$
Finally, using the Cauchy-Schwartz inequality again we have
$$
E\left|V_3\right| \leq
C p E(V_1)
\leq C\ell/(qn^{-\eta})\to 0.
$$

Second, for $\sqrt{\ell}/q\sum_{k=0}^{q-1}[\bar\eps,\bar\eps]_t^{\mathcal{S}_k}$, following the way the terms of $2\sqrt{\ell}/q\sum_{k=0}^{q-1}[X,\bar\eps]_t^{\mathcal{S}_k}$ were rearranged, we have
\begin{align}
\frac{\sqrt{\ell}}{q}\sum_{k=0}^{q-1}[\bar\eps,\bar\eps]_t^{\mathcal{S}_k}
&=\frac{\sqrt{\ell}}{pq}\sum_{t_i\leq t}\left(-2\sum_{j=0}^{p-1}\frac{p-j}{p}\eps_{t_{i-q-j}}-2\sum_{j=1}^{p-1}\frac{p-j}{p}\eps_{t_{i-q+j}}+4\sum_{j=1}^{p-1}\frac{p-j}{p}\eps_{t_{i-j}}\right)\eps_{t_i}\notag\\
&~~~+\frac{2\sqrt{\ell}}{pq}\sum_{t_i\leq t}\eps^2_{t_i}+o_p(1) \notag\\
:&=M^{(4)}_t + \frac{2\sqrt{\ell}}{pq}\sum_{t_i\leq t}\eps^2_{t_i}+o_p(1).\notag
\end{align}
where the $o_p(1)$ term is again due to the end effect.

We first deal with $M^{(4)}_t$. We need the following notation $$J_1:=\{1,2,...,p-1\},~J_2:=\{q-p+1,q-p+2,\ldots,q-1\}~{\rm and}~J_3:=\{q,q+1,\ldots,q+p-1\}.$$ Let $J:=\bigcup_{m=1}^3J_m$
and $J_{\max}$ be the largest element in $J$. Moreover, denote the following weight function
$$
w(j)=\left\{\begin{array}{lll}
4\frac{p-j}{p}~~~{\rm for}~j\in J_1;\\
-2\frac{p-q+j}{p}~~~{\rm for}~j\in J_2;~~~{\rm and}\\
-2\frac{p+q-j}{p}~~~{\rm for}~j\in J_3.
\end{array}
\right.
$$
Notice that $|w(j)|\leq 4$ for all $j\in J$.
$M^{(4)}_t$ is a martingale with quadratic variation that can then be represented as
\begin{align}
\langle M^{(4)},M^{(4)}\rangle_t
&=\frac{\ell\sigma_\eps^2}{p^2q^2}\sum_{t_i\leq t}\left[-2\sum_{j=0}^{p-1}\frac{p-j}{p}\eps_{t_{i-q-j}}-2\sum_{j=1}^{p-1}\frac{p-j}{p}\eps_{t_{i-q+j}}+4\sum_{j=1}^{p-1}\frac{p-j}{p}\eps_{t_{i-j}}\right]^2\notag\\
&=\frac{24\ell\sigma_\eps^2}{p^4q^2}\sum_{t_i\leq t}\sum_{j=1}^{p-1}\left(p-j\right)^2\sigma_\eps^2+\frac{\ell\sigma_\eps^2}{p^2q^2}\sum_{t_i\leq t}\sum_{j\in J}w(j)^2(\varepsilon_{t_{i-j}}^2-\sigma_\varepsilon^2)\notag\\
&~~~+\frac{\ell\sigma_\eps^2}{p^2q^2}\sum_{t_i\leq t}\sum_{j,k\in J,j\neq k}w(j)w(k)\varepsilon_{t_{i-j}}\varepsilon_{t_{i-k}}\notag\\
&=\frac{24\ell\sigma_\eps^2}{p^4q^2}\sum_{t_i\leq t}\sum_{j=1}^{p-1}\left(p-j\right)^2\sigma_\eps^2+o_p(1)\notag\\
&\sim\frac{8n\ell(\sigma_\eps^2)^2}{pq^2}\frac{N_t}{n}+o_p(1)\overset{p}
\rightarrow\frac{8(\sigma_\eps^2)^2C_\ell^3}{C_p}\int_0^t r_s ds\notag
\end{align}
where the last line follows from the assumption that $L_t/\ell\overset{p}\rightarrow\int_0^tr_sds$ and the third equality is explained as follows. We take the third term on the RHS of the second equality for example while the second term can be treated more easily by a similar argument. Notice that this term can be rewritten as
\begin{align}
\mathcal{E}_t:=\frac{2\ell\sigma_\eps^2}{p^2q^2}\sum_{t_i< t^*}\left[\sum_{j=1}^{J_{\max}-1}\left(\sum _{k=1}^{J_{\max}-j}w(k)w(j+k)I_{\{k\in J,~j+k\in J\}}\right)\varepsilon_{t_{i-j}}\right]\varepsilon_{t_i}\notag,
\end{align}
where $t^*:=\max\{t_j\leq t\}.$
Hence by Lemma \ref{lem:rnd_sum}, the BDG inequality, the boundedness of $w(\cdot)$ function and the fact that the cardinality of set $J$ is of order $p$, we have
\begin{align}
E\left(\mathcal{E}_t^2\right)&\leq \frac{4\ell^2\sigma_\eps^6}{p^4q^4}E\sum_{t_i< t^*}\left[\sum_{j=1}^{J_{\max}-1}\left(\sum _{k=1}^{J_{\max}-j}w(k)w(j+k)I_{\{k\in J,~j+k\in J\}}\right)\varepsilon_{t_{i-j}}\right]^2\notag\\
&\leq C\frac{nqp^2\ell^2\sigma_\eps^8}{p^4q^4}=C\frac{n\ell^2}{p^2q^3}=o(1)\notag.
\end{align}
Hence $\mathcal{E}_t=o_p(1)$.

Therefore, based on our moment assumption for $(\eps_{t_i})_{i\geq1}$, $M^{(4)}$ satisfies a CLT where the limiting distribution is a mixture of normal and the mixture component is the variance equal to $\frac{8(\sigma_\eps^2)^2C_\ell^3}{C_p}\int_0^tr_sds$; in other words,
\begin{equation}\label{conv:M_4}
M^{(4)}_t
\Longrightarrow \int_0^t\left[\frac{8(\sigma_\eps^2)^2C_\ell^3}{C_p} r_s \right]^{1/2}dB_s,
\end{equation}
where $B_t$ is a standard Brownian motion that is independent of $\mathcal{F}_1$.

As to $2\sqrt{\ell}/(pq)\sum_i\eps^2_{t_i},$ it follows from Lemma \ref{Lem:ConvgcOfsigeps} and C(4) that
\begin{align}
&~~~\frac{2\sqrt{\ell}}{pq}\sum_{t_i\leq t}\eps^2_{t_i}-\frac{2\sqrt{\ell}}{pq}N_t\widehat{\sigma^2_\eps}\notag\\
&=\frac{2\sqrt{\ell}}{pq}\sum_{t_i\leq t}(\eps^2_{t_i}-\sigma^2_\eps)
-\frac{2\sqrt{\ell}}{pq}N_t\left(\widehat{\sigma^2_\eps}-\sigma^2_\eps\right)\notag\\
&=O_p\left(\frac{\sqrt{\ell n}}{pq}\right)=o_p(1)~{\rm in} ~D[0,1].\notag
\end{align}

\subsection*{Step 3}
Finally, we prove a CLT for term $III$. We have
\[
\overline{M}_t
:=\sqrt{\ell}\left(\frac{1}{q}\sum_{k=0}^{q-1}[X,X]^{\mathcal{S}_k}_t-\langle X,X\rangle_t\right)
=\frac{1}{q}\sum_{k=0}^{q-1}\sqrt{\ell}\left([X,X]^{\mathcal{S}_k}_t-\langle X,X\rangle_t\right)
=\frac{1}{q}\sum_{k=0}^{q-1}M^k_t,\notag
\]
where
$$
dM^k_t=2\sqrt{\ell}(X_t-X_{t_*^k})dX_t
$$
and $t_*^k$ is the largest time smaller than or equal to $t$ on the $k$th sub-grid.
Therefore
\[
\overline{M}_t
=\frac{1}{q}\sum_{k=0}^{q-1}\int_0^tdM^k_s
=\frac{2\sqrt{\ell}}{q}\sum_{k=0}^{q-1}\int_0^t(X_s-X_{t^k_*})dX_s
=2\sqrt{\ell}\int_0^tf_n(s)dX_s,
\]
where
$$
f_n(s)=\left\{
\begin{array}{ll}
    0,  ~~~{\rm for}~s\in[0,t_p);\\
    \frac{1}{q}\sum_{j=0}^{i-p}(X_s-X_{t_{i-j}}),  ~~~{\rm for~} s\in[t_i,t_{i+1})~{\rm and~}p\leq i<q+p; \\
    X_s-X_{t_i}+\sum_{j=1}^{q-1}\frac{q-j}{q}\Delta X_{t_{i-j+1}},~~~{\rm for~}s\in[t_i,t_{i+1})~{\rm and~}i\geq q+p.
\end{array}
\right.  $$
$\overline{M}_t$ is a martingale with quadratic variation
$
\langle \overline{M},\overline{M}\rangle_t=4\ell\int_0^tf_n(s)^2\sigma_s^2 ds.
$
Since $Ef_n(s)^2\leq Cq/n^{1-\eta}$,
\begin{align}
4\ell\int_0^{t_{q+p}}f_n(s)^2\sigma^2_sds
=o_p(1),\label{eq:atypeofendeffect}
\end{align}
Hence, we need to only consider $s\geq t_{q+p}$, i.e., $i\geq q+p.$ By It\^o's formula,
$$
df_n(s)^4=4f_n(s)^3dX_s+6f_n(s)^2\sigma^2_sds, ~{\rm for}~s\in[t_i,t_{i+1})~{\rm and}~i\geq q+p.
$$
Hence
\begin{equation}\label{eq:MMquadraticvariationDecomp}
\langle \overline{M},\overline{M}\rangle_t
=4\ell\int_0^tf_n(s)^2\sigma_s^2ds
=\frac{2\ell}{3}\int_0^t df_n(s)^4
-\frac{8\ell}{3}\int_0^t f_n(s)^3 dX_s.
\end{equation}

We first prove the second term on the RHS of (\ref{eq:MMquadraticvariationDecomp}) is negligible. In fact,
by the BDG inequality,
\begin{align}
Ef_n(s)^6\leq C(q/n^{1-\eta})^3\label{eq:f_nBDGinwquality}
\end{align}
uniformly in $s$. Hence
\begin{align}
E\left(\int_0^tf_n(s)^3dX_s\right)^2&\leq E\left\langle \int_0^\cdot f_n(s)^3dX_s,\int_0^\cdot f_n(s)^3dX_s\right\rangle_t\notag\\
&\leq E\int_0^1f_n(s)^6\sigma_s^2ds\notag\\
&\leq\sigma^2_+\int_0^1Ef_n(s)^6ds\notag\\
&\leq C\int_0^1\left(\frac{q}{n^{1-\eta}}\right)^3ds=O\left(\frac{q^3}{n^{3-3\eta}}\right)\label{eq:BDGunirule},
\end{align}
and
$
\ell\int_0^tf_n(s)^3dX_s
\overset{p}\longrightarrow 0,~{\rm in}~D[0,1],~{\rm as~}n\rightarrow \infty.
$
Now we deal with the first term in (\ref{eq:MMquadraticvariationDecomp}). We shall only focus on the integral on $[t_{i_{q+p}}, t_*]$ where
$t_*$ is the largest $t_i\leq t$; the remainder term is negligible. We then have
\begin{align}
\ell \int_{t_{i_{q+p}}}^{t_*} \ df_n(s)^4
&=\ell\sum_{t_{i}\leq t}\left\{\left(\sum_{j=0}^{q-1}\frac{q-j}{q}\Delta X_{t_{i-j}}\right)^4-\left(\sum_{j=1}^{q-1}\frac{q-j}{q}\Delta X_{t_{i-j}}\right)^4\right\}\notag\\
&=\underbrace{\ell\sum_{t_{i}\leq t}(\Delta X_{t_i})^4}_{\widetilde{I}}+\underbrace{\ell\sum_{t_{i}\leq t}6(\Delta X_{t_i})^2\left(\sum_{j=1}^{q-1}\frac{q-j}{q}\Delta X_{t_{i-j}}\right)^2}_{\widetilde{II}}\notag\\
&+\underbrace{\ell\sum_{t_{i}\leq t}4(\Delta X_{t_i})^3\left(\sum_{j=1}^{q-1}\frac{q-j}{q}\Delta X_{t_{i-j}}\right)}_{\widetilde{III}}+\underbrace{\ell\sum_{t_{i}\leq t}4\Delta X_{t_i}\left(\sum_{j=1}^{q-1}\frac{q-j}{q}\Delta X_{t_{i-j}}\right)^3}_{\widetilde{IV}}\notag.
\end{align}
By the BDG inequality and Lemma \ref{lem:rnd_sum},
$
\widetilde{I}=O_p(\ell/n^{1-2\eta})=o_p(1).
$
Moreover, for term $\widetilde{IV}$, by computing its quadratic variation and using (\ref{eq:f_nBDGinwquality}) we get
$$
\widetilde{IV}=O_p(\ell q^{3/2}/n^{3/2-2\eta})=O_p(1/(\sqrt{\ell}n^{-2\eta}))=o_p(1).
$$
As to term $\widetilde{III}$, we treat it in the same fashion as above by defining for $s\in[t_{i-1},t_{i})$,
\[\aligned
\breve{X}_s&:=4\ell(X_s-X_{t_{i-1}})^3\left(\sum_{j=1}^{q-1}\frac{q-j}{q}\Delta X_{t_{i-j}}\right);\\
\breve{X}^{(1)}_s&:=12\ell(X_s-X_{t_{i-1}})^2\left(\sum_{j=1}^{q-1}\frac{q-j}{q}\Delta X_{t_{i-j}}\right);\\
\breve{X}^{(2)}_s&:=12\ell(X_s-X_{t_{i-1}})\left(\sum_{j=1}^{q-1}\frac{q-j}{q}\Delta X_{t_{i-j}}\right).
\endaligned
\]
Then
\begin{align}
\widetilde{III}=
\sum_{t_{i}\leq t}\int_{t_{i-1}}^{t_{i}}d\breve{X}_s
=\int_0^t\breve{X}^{(1)}_sdX_s+\int_0^t\breve{X}^{(2)}_s\sigma_s^2ds +o_p(1), \notag
\end{align}
which is again an $o_p(1)$ term by noting that (1) $E(\breve{X}^{(1)}_s)^2\leq C\ell^2\cdot 1/n^{2-2\eta}\cdot q/n^{1-\eta} \leq C\ell/n^{2-3\eta}$; and (2)
$E(\breve{X}^{(2)}_s)^2\leq C\ell^2\cdot 1/n^{1-\eta}\cdot q/n^{1-\eta}\leq C\ell/n^{1-2\eta}$.
Finally, by assumption C(6) we get the convergence of term $\widetilde{II}$ and hence
$$
\langle \overline{M},\overline{M}\rangle_t\overset{p}\longrightarrow4\int_0^tw_s\sigma_s^4\ ds~{\rm for~all~}t.
$$

Next, we estimate the quadratic covariation between $\overline{M}$ and $X$. To do so, we first notice that, by It\^o's formula,
\begin{align}
d\langle X,\overline{M}\rangle_t=\frac{1}{q}\sum_{k=0}^{q-1}d\langle X,M^k\rangle_t=\frac{1}{q}\sum_{k=0}^{q-1}\frac{2\sqrt{\ell}}{3}d(X_t-X_{t_*^k})^3-\frac{1}{q}\sum_{k=0}^{q-1}2\sqrt{\ell}(X_t-X_{t_*^k})^2dX_t,\label{eq:cqvXMresidu}
\end{align}
where
$$
\frac{1}{q}\sum_{k=0}^{q-1}\frac{2\sqrt{\ell}}{3}\ d(X_t-X_{t_*^k})^3
=\frac{2}{3}\frac{1}{q}\sum_{k=0}^{q-1}\sqrt{\ell}\ d[X,X,X]_t^{\mathcal{S}_k}.
$$
We next show that the martingale term in (\ref{eq:cqvXMresidu}) is negligible. Rearranging terms the same way as we did for $\overline{M}_t$, we have
\[
R_t:=\int_0^t\frac{1}{q}\sum_{k=0}^{q-1}2\sqrt{\ell}(X_s-X_{t_*^k})^2dX_s\notag
=\frac{2\sqrt{\ell}}{q}\int_0^tg_n(s)dX_s\notag,
\]
where
$$
g_n(s)=\left\{
\begin{array}{ll}
    0,  ~~~{\rm for}~s\in[0,t_p);\\
    \sum_{j=0}^{i-p}(X_s-X_{t_{i-j}})^2,  ~~~{\rm for~} s\in[t_i,t_{i+1})~{\rm and~}p\leq i<q+p; \\
    \sum_{j=0}^{q-1}(X_s-X_{t_{i-j}})^2,~~~{\rm for~}s\in[t_i,t_{i+1})~{\rm and~}i\geq q+p.
\end{array}
\right.  $$
Observe that by the Cauchy-Schwartz inequality and the BDG inequality,
\begin{align}
&~~~E\left[\sum_{j=0}^{q-1}(X_{t_{i+1}}-X_{t_{i-j}})^2\right]^2\notag\\
&=E\sum_{j=0}^{q-1}(X_{t_{i+1}}-X_{t_{i-j}})^4+2E\sum_{0\leq j<k\leq q-1}(X_{t_{i+1}}-X_{t_{i-j}})^2(X_{t_{i+1}}-X_{t_{i-k}})^2\notag\\
&\leq C\sigma_+^4q \frac{q^2}{n^{2-2\eta}}+C\sigma_+^4 q(q-1) \frac{q^2}{n^{2-2\eta}}=O\left(\frac{q^4}{n^{2-2\eta}}\right)\notag.
\end{align}
Hence, uniformly in $s\in[t_i,t_{i+1})$ and $i$,
$
Eg_n(s)^2\leq Cq^4/n^{2-2\eta}\label{eq:QCVgfunctionbound}.
$
Therefore, by the BDG inequality again,
\[
E(R_t)^2
\leq C\frac{\ell}{q^2}E\int_0^tg_n(s)^2\sigma^2_sds\notag
\leq C \frac{\ell q^4}{n^{2-2\eta}q^2}\leq C\frac{1}{\ell n^{-2\eta}}\to 0,\notag
\]
and hence
$R_t=o_p(1).$
Therefore, Assumption C(7) and (\ref{eq:cqvXMresidu}) imply that
$$
\langle X,\overline{M}\rangle_t\overset{p}\longrightarrow\frac{2}{3}\int_0^t\bar{v}_s\sigma_s^3ds~~~{\rm for~all~}t.
$$

It follows from the limit results in either Theorem B.4 (p. 65-67) of
\cite{zhang2001} or Theorem 2.28 of \cite{MZ2012} that, stably in law,
\begin{equation}\label{conv:ol_M}
\overline{M}_t
\Longrightarrow \frac{2}{3}\int_0^t\bar{v}_s\sigma_s dX_s+\int_0^t\left[\left(4w_s-\frac{4}{9}\bar{v}_s^2\right)\sigma_s^4\right]^{1/2}dB_s,
\end{equation}
where $B_t$ is a standard Brownian motion that is independent of $\mathcal{F}_1$.

\subsection*{Step 4}
Clearly $\langle M^4, \overline{M}\rangle_t = 0$. The overall results then follows from \eqref{conv:M_4} and \eqref{conv:ol_M}.

\qed












\end{document}


\begin{frontmatter}



\title{Supplement to ``Volatility Inference in the Presence of Both  Endogenous Time and Microstructure Noise''}


\author[1]{Yingying Li\fnref{label1}}
\address[1]{Department of Information Systems, Business Statistics and Operations Management, Hong Kong University of Science and Technology. Email: yyli@ust.hk}

\author[2]{Zhiyuan Zhang\fnref{label2}\corref{cor1}}
\address[2]{School of Statistics and Management, Shanghai University of Finance and Economics. Email: zhang.zhiyuan@mail.shufe.edu.cn}

\author[3]{Xinghua Zheng\fnref{label3}}
\address[3]{Department of Information Systems, Business Statistics and Operations Management, Hong Kong University of Science and Technology. Email: xhzheng@ust.hk}

\fntext[label1]{Research partially supported
by GRF 602710 of the HKSAR.}
\fntext[label2]{Research partially supported by
Shanghai Pujiang Program 12PJC051.}
\fntext[label3]{Research partially supported
by GRF 602710 and GRF 606811 of the HKSAR.}

\end{frontmatter}

\setcounter{section}{6}
\setcounter{equation}{28}
\renewcommand{\theequation}{\thesection.\arabic{equation}}

This supplementary article is dedicated to prove Proposition 1 in the main article \cite{LZZlama}. The theorem/proposition numbers below refer to the main article.

\begin{Lem}\label{Lem:skewness_conv}
Under the assumptions of Proposition 1, in $D[0,1]$,
$$
\sqrt{\ell}[\overline{Y},\overline{Y},\overline{Y}]^{\mathcal{S}}_{t}\overset{p}\longrightarrow\int_0^{t}v_s\sigma_s^3\ ds~~~{\rm and}~~~\frac{1}{q}\sum_{k=0}^{q-1}\sqrt{\ell}[\overline{Y},\overline{Y},\overline{Y}]^{\mathcal{S}_k}_{t}
\overset{p}\longrightarrow\int_0^{t}\bar{v}_s\sigma_s^3\ ds.
$$
Moreover, in $D[0,1]$,
$$
\left|\frac{1}{q}\sum_{k=0}^{q-1}\sqrt{\ell}[\overline{Y},\overline{Y},\overline{Y}]^{\mathcal{S}_k}_{t}-\int_0^{t}\bar{v}_s\sigma_s^3ds\right|/\delta_{n}\overset{p}\rightarrow0,~~~{\rm for}~~~\delta_{n}\rightarrow0~~~{\rm and}~~~1/\delta_{n}=o\left(n^{\gamma(\alpha,\eta)}\right).
$$
\end{Lem}
\begin{proof} Using the notation in Step 2 of the proof of Theorem 1, we have
\begin{align}
&~~~\sqrt{\ell}\left([\overline{Y},\overline{Y},\overline{Y}]^{\mathcal{S}}_t-[\widetilde{Y},\widetilde{Y},\widetilde{Y}]^{\mathcal{S}}_t\right)\notag\\
&=\underbrace{\sqrt{\ell}\sum_{t_{i,0}\leq t}(\Delta C_i)^3}_{i}+\underbrace{3\sqrt{\ell}\sum_{t_{i,0}\leq t}\Delta A_i(\Delta C_i)^2}_{ii}+\underbrace{3\sqrt{\ell}\sum_{t_{i,0}\leq t}(\Delta A_i)^2\Delta C_i}_{iii}\notag\\
&~~~~~~+\underbrace{3\sqrt{\ell}\sum_{t_{i,0}\leq t}\Delta B_i(\Delta C_i)^2}_{iv}+\underbrace{3\sqrt{\ell}\sum_{t_{i,0}\leq t}(\Delta B_i)^2\Delta C_i}_{v}+\underbrace{6\sqrt{\ell}\sum_{t_{i,0}\leq t}\Delta A_i\Delta B_i\Delta C_i}_{vi}\notag,
\end{align}
and
\begin{align}
\sqrt{\ell}[\widetilde{Y},\widetilde{Y},\widetilde{Y}]^{\mathcal{S}}_t&=\sqrt{\ell}\sum_{t_{i,0}\leq t}(\Delta A_i)^3+\underbrace{\sqrt{\ell}\sum_{t_{i,0}\leq t}(\Delta B_i)^3}_{vii}\notag\\
&~~~+\underbrace{3\sqrt{\ell}\sum_{t_{i,0}\leq t}(\Delta A_i)^2\Delta B_i}_{viii}+\underbrace{3\sqrt{\ell}\sum_{t_{i,0}\leq t}\Delta A_i(\Delta B_i)^2}_{ix}\notag.
\end{align}
Observe that the leading term
$
\sqrt{\ell}\sum_{t_{i,0}\leq t}(\Delta A_i)^3=\sqrt{\ell}[X,X,X]_t^{\mathcal{S}}.
$
Applying the same techniques used in the proof of Theorem 1 or 2, we can show that with the tedious calculation being omitted
\begin{align}
&E|i|\leq C\left(\frac{\ell p}{n^{1-\eta}}\right)^{3/2},~~~E|ii|\leq C\frac{\ell p}{n^{1-3\eta/2}},~~~E|iii|\leq C\sqrt{\frac{\ell p}{n^{1-3\eta}}},\notag\\
&E|iv|\leq C\frac{\ell\sqrt{p}}{n^{1-\eta}},~~~E|v|\leq C\frac{\ell}{\sqrt{pn^{1-\eta}}},~~~E|vi|\leq C\sqrt{\frac{\ell}{n^{1-2\eta}}}\notag,\\
&E|vii|\leq C\frac{\ell}{p^{3/2}},~~~E|viii|\leq C\frac{1}{p^{1/2}n^{-\eta}},~~~{\rm and}~~~E|ix|\leq C\frac{\ell^{1/2}}{pn^{-\eta/2}}\notag.
\end{align}
The first half of the lemma follows by noticing that all bounds above are $o(1)$ given our assumptions. The second half of the lemma follows by the assumption C(7') and a careful examination of the orders of the above bounds.
\end{proof}

\begin{Lem}\label{Lem:F_f}
Suppose that there exists a (nonrandom) sequence $(\delta_n)_{n\geq 1}$ with $\delta_n\rightarrow 0$ such that
$$|F_n(s)-F(s)|/\delta_n\overset{p}\rightarrow 0~{\rm in~}D[0,1],~{\rm where~}F(t):=\int_0^tf(s)ds,$$
and $f(t)$ is continuous a.s.. Suppose further that $(\tau_i)_{i\geq1}$ is a (possibly random) partition over $[0,1]$ with
$\widetilde{\Delta}_n:=\max_i|\tau_{i}-\tau_{i-1}|=o_p(1)$ and
\begin{equation}\label{eq:blocking}
\delta_n/\min_i|\tau_{i}-\tau_{i-1}|=O_p(1).
\end{equation}
Define
\begin{equation}\label{dfn:F_f}
f_n(t):=(F_n(\tau_{i})-F_n(\tau_{i-1}))/(\tau_{i}-\tau_{i-1}),~{\rm for}~t\in[\tau_{i},\tau_{i+1}).
\end{equation}
Then
$
f_n(t)\overset{p}\longrightarrow f(t)~in~D[0,1].
$
\end{Lem}
\begin{proof}
Since we are proving convergence in probability, without loss of generality we can act as if \eqref{eq:blocking} is
reinforced to
\[
\min_i|\tau_{i}-\tau_{i-1}| \geq \frac{\delta_n}{C}.
\]
Now notice that for $t\in[\tau_{i},\tau_{i+1})$ and all $i$
\begin{align}
&~~~~~~|f_n(t)-f(t)|=\frac{1}{\tau_{i}-\tau_{i-1}}|F_n(\tau_{i})-F_n(\tau_{i-1})-f(t)(\tau_{i}-\tau_{i-1})|\notag\\
&\leq\underbrace{\frac{C}{\delta_n}|F_n(\tau_{i})-F_n(\tau_{i-1})-(F(\tau_{i})-F(\tau_{i-1}))|}_{I}\notag\\
&~~~+\underbrace{\frac{1}{\tau_{i}-\tau_{i-1}}|F(\tau_{i})-F(\tau_{i-1})-f(t)(\tau_{i}-\tau_{i-1})|}_{II}\notag.
\end{align}
Since the uniform topology is identical to the Skorokhod topology regarding convergence to a continuous function on $[0,1]$
(see, e.g., pp.328 of \citet{JacodShiryaev}),
we can treat $I$ and $II$ respectively as follows:
\[
\aligned
I&\leq\frac{C}{\delta_n}|F_n(\tau_{i})-F(\tau_{i})|+\frac{C}{\delta_n}|F_n(\tau_{i-1})-F(\tau_{i-1})|
\leq\frac{2C}{\delta_n}\sup_{s\in[0,1]}|F_n(s)-F(s)|\to 0,\q\mbox{and}\\
II&\leq\frac{1}{\tau_{i}-\tau_{i-1}}\int_{\tau_{i-1}}^{\tau_{i}}|f(s)-f(t)|ds\leq C_{n}(f)\to 0,\notag
\endaligned
\]
where $C_{n}(f):=\sup_{|s-t|\leq2\widetilde{\Delta}_n}|f(s)-f(t)|\overset{p}\rightarrow0$, by the \hbox{a.s.} continuity of $f(t)$ on the compact interval $[0,1]$ and that $\widetilde{\Delta}_n=o_p(1)$.
\end{proof}

Observe that we intentionally shift the time when we define $f_n(t)$ in Lemma \ref{Lem:F_f}, so that it is adapted to the filtration $(\mathcal{F}_t).$

\begin{Cor}\label{Cor:conv_spot_vol}
Under the assumptions of Theorem 2 and that $\sigma_t^2$ is a.s. continuous. If $\delta_n$ is chosen such that $1/\delta_n=o(n^{\alpha/2}),$ then $f_n^{(2)}(t)\overset{p}\longrightarrow f^{(2)}(t)$ in $D[0,1].$
\end{Cor}
\begin{proof}
This is a direct consequence of Theorem 2 and Lemma \ref{Lem:F_f}.
\end{proof}

\begin{Lem}\label{Lem:conv_prod_ratio}
Suppose
\[
 f_n(t)\overset{p}\longrightarrow f(t), \q\mbox{and}\q
 g_n(t)\overset{p}\longrightarrow g(t),\q\mbox{both  in }D[0,1],
\]
where $f(t),g(t)$ are continuous a.s.. Moreover, suppose that, almost surely, $|f(t)|\leq b$ and $1/b\leq|g(t)|\leq b$ for some $b>0$ and all $t\in[0,1]$. Then
$$
\frac{f_n(t)}{g_n(t)}\overset{p}\longrightarrow\frac{f(t)}{g(t)}~ in~D[0,1],~and~f_n(t)g_n(t)\overset{p}\longrightarrow f(t)g(t)~ in~D[0,1].
$$
\end{Lem}
\begin{proof}
The proof is straightforward and hence omitted.
\end{proof}

Now we are ready to prove Proposition 1.
\begin{proof}[Proof of Proposition 1]
Lemma \ref{Lem:skewness_conv} and Corollary \ref{Cor:conv_spot_vol} state that Lemma \ref{Lem:F_f} applies for $F^{(3)}_n$ and $F_n^{(2)}$ with $1/\delta_n=o\left(n^{\gamma(\alpha,\eta)}\right)$. This together with Lemma \ref{Lem:conv_prod_ratio} and assumptions made in the proposition imply that
$$
\frac{f^{(3)}_n(t)}{f^{(2)}_n(t)}\overset{p}\longrightarrow\frac{f^{(3)}(t)}{\sigma^2_t}=\bar{v}_t\sigma_t,~{\rm in}~D[0,1].
$$
Therefore, it suffices to show that if $f_n\overset{p}\longrightarrow f~{\rm in}~D[0,1]$ and $f$ is continuous on $[0,1]$ a.s., then
\begin{equation}\label{stoc_int_barY_same}
\sum_{\tau_{i}\leq t}f_n(\tau_{i-1})\Delta \overline{Y}_{\tau_i}\overset{p}\longrightarrow\int_0^tf(s)\ dX_s~{\rm in}~D[0,1].
\end{equation}

To prove \eqref{stoc_int_barY_same}, note that $f_n\overset{p}\longrightarrow f~{\rm in}~D[0,1]$ and the continuity of $f$ imply $\sup_t|f_n(t) - f(t)|\toop~0 $. Hence as we are only proving convergence in probability, we can without loss of generality assume that $\sup_t|f_n(t)|\leq C$ for some $C>0$ and for all $n$.
Next, according to the definition of $\Delta\overline{Y}_{\tau_i}$, the above summation can be rewritten as
$$
\sum_{\tau_{i}\leq t}f_n(\tau_{i-1})\Delta \overline{Y}_{\tau_i}=\underbrace{\sum_{\tau_{i}\leq t}f_n(\tau_{i-1})\Delta \overline{X}_{\tau_i}}_{I(t)}+\underbrace{\sum_{\tau_{i}\leq t} f_n(\tau_{i-1})\Delta\bar{\eps}_{\tau_i}}_{II(t)},
$$
where
\[\aligned
\Delta\overline{X}_{\tau_i}&:=\frac{1}{p}\sum_{j=0}^{p-1}X_{t_{id_1q+p-j}}-\frac{1}{p}\sum_{j=0}^{p-1}X_{t_{(i-1)d_1q+p-j}},\\ \Delta\bar{\eps}_{\tau_i}&:=\frac{1}{p}\sum_{j=0}^{p-1}\eps_{t_{id_1q+p-j}}-\frac{1}{p}\sum_{j=0}^{p-1}\eps_{t_{(i-1)d_1q+p-j}}.
\endaligned\]
Firstly, thanks to the Doob's $L^p$ inequality in Appendix A.1 in \cite{LZZlama}, since
\begin{align}
E(II(1)^2)&\leq 4E\sum_{\tau_{i}\leq 1}f_n(\tau_{i-1})^2\frac{\sigma^2_\eps}{p}\leq C\frac{n}{pd_1q}=o(1),\notag
\end{align}
$II(t)$ is $o_p(1)$ in $D[0,1]$.
Secondly,
$$
I(t)=\sum_{\tau_{i}\leq t}f_n(\tau_{i-1})(X_{t_{\tau_i}}-X_{t_{\tau_{i-1}}})+
\sum_{\tau_{i}\leq t}f_n(\tau_{i-1})\left(\sum_{j=1}^{p}\frac{p-j+1}{p}\Delta X_{t_{id_1q+j}}
- \sum_{j=1}^{p}\frac{p-j+1}{p}\Delta X_{t_{(i-1)d_1q+j}}\right).
$$
By Theorem VI.6.22 (c) of \citet{JacodShiryaev}, the first term in the RHS of the above equation converges in probability (in $D[0,1]$) to the target $\int_0^tf(s)\ dX_s$.
Finally, the second term above is $o_p(1)$ in $D[0,1]$ due to that
\[\aligned
E\left(\sum_{\tau_{i}\leq t}f_n(\tau_{i-1})\sum_{j=1}^{p}\frac{p-j+1}{p}\Delta X_{t_{id_1q+j}}\right)^2
&\leq C\frac{1}{n^{1-\eta}}E\sum_{\tau_{i}\leq t}f_n(\tau_{i-1})^2 \sum_{j=1}^{p}\left(\frac{p-j+1}{p}\right)^2 \\
&\leq C\frac{p}{d_1qn^{-\eta}}=o(1),
\endaligned\]
and similarly for the other summand.
\end{proof}